%% file: allocation.tex
\documentclass[a4paper,UKenglish]{article}

\include{preamble}

\begin{document}

\title{\textbf{PTAS for Ordered Instances of Resource Allocation Problems with Restrictions on Inclusions}}

\author[1]{Kamyar Khodamoradi\thanks{The author wishes to acknowledge support from a Natural Sciences and Engineering Research Council of Canada (NSERC)  Discovery Grant.}}
\author[1]{Ramesh Krishnamurti\thanks{The author wishes to acknowledge support from a Natural Sciences and Engineering Research Council of Canada (NSERC)  Discovery Grant.}}
\author[1, 2]{Arash Rafiey \thanks{The author is supported by COMPETE grant from Indiana State University}}
\author[3]{Georgios Stamoulis\thanks{The work of the author was partially supported by the Swiss National Foundation as part of the \textit{Early Postdoc.Mobility} grant P1TIP2\_152282}}
\affil[1]{School of Computing Science, Simon Fraser University, Burnaby, BC, Canada}
\affil[2]{Department of Math and Computer Science, Indiana State University, Terre Haute, IN, USA}
\affil[3]{Depatrment of Data Science \& Knowledge Engineering, Maastricht University, The Netherlands}
\affil[ ]{}
\affil[ ]{\textit{kka50@sfu.ca, ramesh@sfu.ca}}
\affil[ ]{\textit{arash.rafiey@indstate.edu}} 
\affil[ ]{\textit{georgios.stamoulis@maastrichtuniversity.nl}}
\renewcommand\Authands{ and }

\maketitle

\begin{abstract}
We consider the problem of allocating a set $I$ of $m$ indivisible resources (items) to a set $P$ of $n$ customers (players) competing for the resources. Each resource $j \in I$ has a same value $v_j > 0$ for a subset of customers interested in $j$, and zero value for the remaining customers. The utility received by each customer is the sum of the values of the resources allocated to her. The goal is to find a feasible allocation of the resources to the interested customers such that for the Max-Min allocation problem (Min-Max allocation problem) the minimum of the utilities (maximum of the utilities) received by the customers is maximized (minimized). The Max-Min allocation problem is also known as the \textit{Fair Allocation problem}, or the \textit{Santa Claus problem}. The Min-Max allocation problem is the problem of Scheduling on Unrelated Parallel Machines, and is also known as the $R \, | \, | C_{\max}$ problem.

In this paper, we are interested in instances of the problem that admit a Polynomial Time Approximation Scheme (PTAS). We show that an ordering property on the resources and the customers is important and paves the way for a PTAS. For the Max-Min allocation problem, we start with instances of the problem that can be viewed as a \textit{convex bipartite graph}; a bipartite graph for which there exists an ordering of the resources such that each customer is interested in (has a positive evaluation for) a set of \textit{consecutive} resources. We demonstrate a PTAS for the inclusion-free cases. This class of instances is equivalent to the class of bipartite permutation graphs. 
For the Min-Max allocation problem, we also obtain a PTAS for inclusion-free instances. These instances are not only of theoretical interest but also have practical applications.

\end{abstract}

\section{Introduction and Problem Definition}

In the general resource allocation problem, we are given a bipartite graph $H = (I, P ,E)$ where the set of vertices $I$ represents $m$ indivisible resource items (or simply items) to be allocated to the set $P$ of $n$ players (or customers). The sets $I$ and $P$ form the two sets of the bipartition, and are  indexed by numbers in $[m]$ and $[n]$ respectively. Thus, $I = \{ x_1, \, x_2, \, \ldots, \, x_m \}$ and $P = \{ p_1, \, p_2, \, \ldots, \, p_n \}$. Each player $p_j \in P$ has a utility function $val_j(i) = v_{i} \geq 0$ for each item $x_i \in I$. This represents the value of item $x_i$ for player $p_j$. If item $x_i$ is adjacent to player $p_j$, then its value for her is $v_i>0$ ($v_i$ is a positive integer), otherwise its value is zero. This is represented in the graph $H$ via the edge set $E$. If an item $x_i$ has a value $v_i > 0$ to a player $p_j$, then there exists an edge $e \in E$ that connects $x_i$ to $p_j$ in $H$, and we say player $p_j$ is adjacent to item $x_i$. If there is no edge between an item $x_i \in I$ and a player $p_j \in P$, then the item $x_i$ has a value of zero to  player $p_j$. Allocating an item $x_i \in I$ with value $v_i \in \bQ_{>0}$, where $\bQ_{>0}$ is the set of positive rational numbers, to a player adjacent to the item increases the utility of that player by $v_i$. Our goal is to find a feasible allocation of the resource items to the players that optimizes an objective function. 

We assume there are no items of degree zero in $H$. If there exists such an item, we can safely remove it from the graph since it is not accessible to any player. In a feasible allocation (henceforth, \emph{feasible assignment}) for $n$ players, the set $I$ is partitioned into $n$ subsets, where each subset is allocated to the player with the same index. In other words, $I = I^1 \cup I^2 \cup \dots \cup I^n$, where items in $I^j$ are allocated to player $p_j \in P$ who is adjacent to these items. Such a partition is denoted by $\bcA = (I^1, I^2, \ldots, I^n)$.

The objective function that is optimized depends on the particular resource allocation problem we solve. In the Max-Min allocation problem, we seek a feasible solution that ensures that each player receives utility \emph{at least} $t>0$, assuming that we can guess the largest possible value of the parameter $t$. In other words, we seek a partition $\bcA = (I^1, \dots, I^n)$ such that $\sum_{i: x_i \in I^j} v_i \geq t$, $\forall j \in [n]$. A sufficiently good guess of $t$ can be obtained by doing a binary search. It is worth mentioning that the assignments made in the Max-Min allocation problem do not necessarily need to be a partitioning of the items, meaning that one can leave some items unassigned. But since the goal is to maximize the sum of item values assigned, restricting the assignments to cover the entire set of items will not decrease the objective value. Therefore, to avoid further complications, we assume the assignments made are in fact partitions of the entire set of items. On the other hand, in the Min-Max allocation problem, we have item costs rather than item values. For instance, $v_i$'s can be the construed as processing times of a set of jobs and players can be processors, and we are required to assign \underline{every} job to some processor. Naturally, we wish to allocate jobs to machines such that the overall finish time, or \emph{make-span}, is minimized. Therefore, we seek a feasible solution that ensures each player receives a utility of \emph{at most} $t>0$ for the smallest $t$ possible. In other words, we seek a partition $\bcA = (I^1, \dots, I^n)$ such that $\sum_{i: x_i \in I^j} v_i \leq t$, $\forall j \in [n]$.

In this paper, we study cases where we can obtain PTAS for these two resource allocation problems. We start with cases where the bipartite graph $H = (I, \, P, \, E)$ that models the problem is \textit{convex}. A bipartite graph $H = (I, \, P, \, E)$ with two set of vertices $I$ and $P$, is {\em convex} if there is an ordering $<_I$ of the vertices in $I$ such that the neighbourhood of each vertex in $P$ consists of consecutive vertices, \emph{i.e.}, the neighbourhood of each vertex in $P$ forms an interval. Indeed, we focus on \emph{inclusion-free} instances which form a subset of the convex bipartite graphs. This subclass of graphs is also referred to as the class of \emph{bipartite permutation graphs} and the class of \emph{proper interval bigraphs}. We will discuss them in more detail in \secref{santa_prelim}. Convex bipartite graphs (henceforward referred to as \emph{convex graphs} for short) were introduced in \cite{glover_convex} and are well known for their nice structures and both theoretical and practical properties. Many optimization problems become polynomial-time solvable or even linear-time solvable in convex graphs while remaining {\NPH} for general bipartite graphs  \cite{Brandstadt:1999:GCS:302970}. Convex graphs can be recognized in linear time by using PQ-trees \cite{Brandstadt:1999:GCS:302970} and, moreover, the recognition algorithms provide the ordering $<_I$, given that the underlying graph is of course convex \cite{DBLP:journals/dam/MeidanisPT98, DBLP:journals/tcs/HabibMPV00, DBLP:journals/jcss/BoothL76}. Various studies have considered the bipartite permutation graphs for other types of optimization problems. This class is of interest in graph homomorphism problems \cite{GHRY08}. Also, there are recognition algorithms for bipartite permutation graphs \cite{GHRY08, SBS87}.

The interval case arises naturally in energy production applications where resources (energy) can be assigned and used within a number of successive time steps \cite{KKRS13, MS12}. That is, the energy produced at some time step is available only for a limited period corresponding to an interval of time steps. The goal is a fair allocation of the resources over time, \emph{i.e.}, an allocation that maximizes the minimum accumulated resource we collect at each time step. In other words, we would like to have an allocation that guarantees the energy we collect at each time step is at least $t$, a pre-specified threshold. See also \cite{DBLP:journals/jal/Sgall96} for some applications in on-line scheduling.

\subsection{Related Work}

For the general Max-Min fair allocation problem, where a given item does not necessarily have the same value for each player (i.e., every player has her own value for any item), no ``good'' approximation algorithm is known. In \cite{BezakovaD05}, by using the natural LP formulation for the problem and similar ideas as in \cite{LenstraST87}, an additive ratio of $\max_{i,j} v_{ij}$ is obtained, which can be arbitrarily bad. A stronger LP formulation, the \textit{configuration LP}, is used to obtain a solution at least $\textrm{opt}/n$ in \cite{BansalS06}. Subsequently, \cite{sadpourS07} provided a rounding scheme for this LP to obtain a feasible solution of value no worse than $\mathcal{O}(\frac{\textrm{opt}}{\sqrt{n}(\log^3 n)})$. Recently in \cite{innovations/SahaS10}, an $\mathcal{O}(\sqrt{\frac{\log \log n}{n \log n}})$ approximation algorithm was proposed, which is close to the \textit{integrality gap} of the configuration LP. On the negative side, there exists a simple $\frac{1}{2}$ inapproximability result \cite{BezakovaD05} using the same ideas as in \cite{LenstraST87}. Due to the difficulty  of the general case, subsequent research focused more on special cases.  For the \textit{restricted} case of the Max-Min allocation problem (also known as the restricted Santa Claus problem), where $v_{ij} \in \{0, v_j \}$ for $i \in [n]$ and $j \in [m]$, there is an $\mathcal{O}(\frac{\log \log \log n}{\log \log n})$ factor approximation algorithm \cite{BansalS06}. Furthermore, the $\frac{1}{2}$ inapproximability result for the general case \cite{BezakovaD05} carries over to this restricted case as well. 

Recently, Feige proved that the integrality gap of the configuration LP defined and studied in \cite{BansalS06}, cannot be worse than some constant. In \cite{AsadpourFS08} an integrality gap of $\frac{1}{5}$ was shown for the same LP which was later improved to $\frac{1}{4}$. This implies that a $\frac{1}{4}$-approximation algorithm based on rounding the configuration LP for the restricted Max-Min allocation is possible, although no such algorithm is available yet. 

The authors provide a local search heuristic with an approximation guarantee of $\frac{1}{4}$. However, this heuristic was not known to run in polynomial time.
Later, it was shown in \cite{svensson_icalp_santa} that the local search can be done in $n^{\mathcal{O}(\log n)}$ time. In \cite{srinivasan} the authors provided a constructive version of Feige's original nonconstructive argument based on Lov\'{a}sz Local Lemma, thus providing the first constant factor approximation for the restricted Max-Min fair allocation problem. They provide an $\alpha$-approximation algorithm for some \textit{constant} $\alpha$ where an explicit value of $\alpha$ is not provided (but is thought to be a huge constant).
Thus there is still a gap between the $\frac{1}{2}$ inapproximability result and the constant $\alpha$ approximability result in \cite{srinivasan}.
Very recently, a $13$-approximation was given for the problem \cite{DBLP:journals/corr/AnnamalaiKS14}, which provides the first constant factor polynomial time approximation algorithm for the problem for a particular constant value. Their approach uses, in a highly non-trivial way, the local search technique for hypergraph matching that was used in \cite{AsadpourFS08}. Another important aspect of this approach is that the algorithm is purely combinatorial.

Several special cases of the Max-Min fair allocation problem have been studied. The case where $v_{ij} \in \{ 0,1, \infty \}$ is shown to be hard in  \cite{approx/KhotP07} and a trade-off between the running time and the approximation guarantee is established. In particular, for a number $\alpha \leq \frac{|I|}{2}$, it is shown how to design an $\alpha \cdot \frac{opt}{n}$-approximation algorithm in time $|P|^{\mathcal{O}(1)}|I|^{\mathcal{O}(\alpha)}$. In \cite{stoc/BateniCG09} the authors consider the case in which each item has positive utility for a bounded number of players $D$, and prove that the problem is as hard as the general case for $D \leq 3$. They also provide a $\frac{1}{2}$ inapproximability result and a $\frac{1}{4}$ approximation algorithm for the \textit{asymmetric} case (where each item has two distinct non-zero values for the two players interested in that item) when $D \leq 2$. The authors also provide a simpler LP formulation for the general problem and devise a polylogarithmic approximation algorithm that runs in quasipolynomial time (and a $|P|^{\epsilon}$-approximation algorithm that runs in $|P|^{\mathcal{O}(\frac{1}{\epsilon})}$ time, for some $\epsilon \geq 0$). The same result has been obtained in \cite{focs/ChakrabartyCK09}, which includes a $\frac{1}{2}$ approximation when $D \leq 2$, thus matching the bound proved in  \cite{stoc/BateniCG09}. In \cite{woeginger_ptas}, the author provides a PTAS for a (very) special case of the problem considered in this paper, namely, when the instance graph of the problem is a complete bipartite graph. In \cite{MS12} a $\frac{1}{2}$-approximation algorithm was developed for the class of instances considered in this paper, namely for the case where the intervals for each player are inclusion-free in the same sense introduced in this paper. See also \cite{DBLP:journals/corr/abs-1009-4529}, \cite{DBLP:journals/orl/MuratoreSW10} for some other special cases that our results generalize.

The Min-Max-Allocation problem, or the $R \, | \, | C_{\max}$ problem, as it is known in standard scheduling notation, is an important class of resource allocation problems.  We seek an allocation (also known as an assignment) of jobs (the resources) to machines (the players) such that the makespan (the time by which the latest machine finishes its processing) is minimized.  

For the $R \, | \, | C_{\max}$ problem, a $2$-approximation algorithm based on a characterization of the extreme point solutions of the natural linear programming relaxation of the problem is given \cite{LenstraST87}. The authors also provide a $\frac{3}{2}$ inapproximability result. These results have been adapted to the Max-Min case in \cite{BezakovaD05}. So far, all efforts to improve either of the bounds have failed. In a very recent result  \cite{rcmax_ola}, it is shown that the restricted version of $R \, | \, | C_{\max}$ (where the processing time of a job $j$ is $v_j$ for a subset of the machines and infinity otherwise) admits an $\alpha$ approximation guarantee for $\alpha$ strictly less than $2$. This result is an \textit{estimation} result i.e., it estimates the (optimal) makespan of the restricted $R \, | \, | C_{\max}$ within a factor of $\alpha = \frac{33}{17} + \epsilon$ for some arbitrary small positive $\epsilon$, although no polynomial time algorithm with such a performance guarantee is known. In \cite{DBLP:journals/algorithmica/EbenlendrKS14}, a $1.75$ approximation algorithm is given for the restricted $R \, | \, | C_{\max}$ problem where each job can be assigned to at most $2$ machines. In this article, the authors claim that their $1.75$ approximation for this restricted case is the first one that improves the factor $2$ approximation of \cite{LenstraST87} on an unbounded number of machines and jobs. Our PTAS thus provides a certain strengthening of their claim, providing the fist natural and non-trivial instance (as in the case of a complete bipartite graph) for which a PTAS is provided.

We note that further restrictions of this special case, where every job has a degree at most two, have been studied \cite{DBLP:journals/ipl/LeeLP09a}. Thus if the underlying bipartite graph is a tree, then a PTAS can be designed for the problem. If the processing times are in the set $\in \{1,2 \}$, then a $\frac{3}{2}$-approximation algorithm exists, which is the best possible unless $\mathcal{P} = \mathcal{NP}$. Finally, \cite{DBLP:journals/scheduling/VerschaeW14} provides better approximations for several special cases of the problem. More importantly, it shows that the configuration LP for this restriction has an integrality gap of $2$.

\subsection{Our Contribution} 

We summarize our results below:

\begin{enumerate}
\item We start by presenting a PTAS for the restricted Max-Min fair allocation problem when the instance of the problem is an inclusion-free convex graph, that is, the neighborhood of each player is an interval of items (\secref{max-min}). Notice that this instance of the problem is (strongly) {\NPH}, as it contains complete bipartite graphs as a special case (each player is adjacent to all the items), which is known to be strongly {\NPH} \cite{garey_johnson}. 

\item Next, we modify our approach for the Max-Min allocation problem to obtain a PTAS for the $R \, | \, | \, C_{\max}$ problem with inclusion-free convex graphs. (\secref{rc_max}).
\end{enumerate}

To obtain the PTAS for the instances considered in this paper, we first use scaling and rounding to classify the items into small and big items. Then, we prove that in a given instance of the problem, for any assignment of a certain value, another assignment of slightly less value but with simpler structure exists. Finally, we provide an algorithm that searches all these simple-structured assignments.

\section{Max-Min Allocation Problem (Santa Claus) on Convex Graphs}
\label{sec:max-min}

As instance of the problem is given via a convex graph $H = (I, \, P, \, E)$ and a utility function $val: \, I \rightarrow \bQ_{>0}$ which associates a value $v_i$ to every item $x_i$ in $I$. Structural properties and algorithms of convex graphs have received attention in the field of algorithmic graph theory. We refer the reader to texts such as \cite{Brandstadt:1999:GCS:302970} and \cite{Spinrad2003}. For the sake of completeness, we give a definition of this class of graphs and introduce some notations that we will use throughout the paper.

\subsection{Preliminaries and Notations}
\label{sec:santa_prelim}

To define the class of convex graphs we first need to introduce the following property.

\begin{definition}[Adjacency Property]\cite{Brandstadt:1999:GCS:302970}
Let $H = (X, \, Y, \, E)$ be a bipartite graph. An ordering $<_X$ of $X$ in $H$ has the \emph{adjacency property} if for each vertex $y \in Y$, the neighbourhood of $y$ in $X$ denoted by $N(y)$ consists of vertices that are consecutive in $<_X$.
\end{definition}

\noindent A convex graph is defined as follows. 

\begin{definition}[Convex Graph]\cite{Brandstadt:1999:GCS:302970}
A bipartite graph $H = (X, \, Y, \, E)$ is convex if there exists an ordering of $X$ or $Y$ that satisfies the adjacency property.
\end{definition}

In this paper, we deal with a subclass of convex graphs known as inclusion-free convex graphs. In this subclass, inclusion may still occur between the intervals, but it should follow certain rules. We first explain the rules with the help of \figref{intersect_intervals}. Assume that a convex graph $H = (I, \, P, \, E)$ is given as the input instance. Consider two players $p$ and $q$ in $P$. Their neighbourhoods in $H$ can either be separate (their intersection is empty) or intersecting. If the neighbourhood of two players $p$ and $q$ is intersecting, four types of intersection between the two neighbourhoods may occur. These types are depicted in \figref{intersect_intervals}. We say that two intervals are \emph{properly overlapping} (\figref{intersect_intervals} \textbf{(a)}) if neither of the two intervals contains the other one, or \emph{left inclusive} (\figref{intersect_intervals} \textbf{(b)}) if one interval contains the other and the two intervals share their left boundary, or 3) \emph{right inclusive} (\figref{intersect_intervals} \textbf{(c)} if one interval contains the other and the two intervals share their right boundary, or 4) \emph{margined-inclusive} (\figref{intersect_intervals} \textbf{(d)}) if one interval is completely contained in the other, and the intervals do not share their boundaries. The subclass of inclusion-free convex graphs forbids margined-inclusion while left inclusion and right inclusion may still occur. 

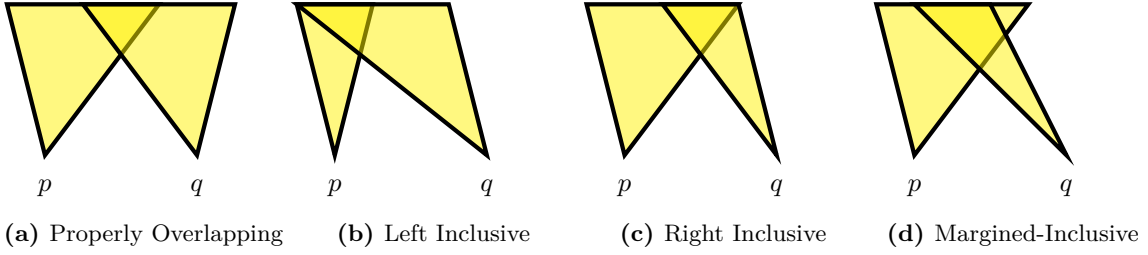
\begin{figure}
    \centering
    \begin{subfigure}[b]{0.22\textwidth}
        \begin{tikzpicture}
            \draw [ultra thick, draw=black, fill=yellow, fill opacity=0.5]    
                   (-0.5, 0) -- (-1, 2) -- (1, 2) -- cycle;
            \draw [ultra thick, draw=black, fill=yellow, fill opacity=0.5]
                   (1.5, 0) -- (0, 2) -- (2, 2) -- cycle;                
            \node[below = 0.2cm] at (-0.5, 0) {$p$};
            \node[below = 0.2cm] at (1.5, 0) {$q$};
        \end{tikzpicture}
        \caption{Properly Overlapping}
    \end{subfigure}
       \begin{subfigure}[b]{0.22\textwidth}
        \begin{tikzpicture}
            \draw [ultra thick, draw=black, fill=yellow, fill opacity=0.5]
                   (-0.5, 0) -- (-1, 2) -- (0, 2) -- cycle;
            \draw [ultra thick, draw=black, fill=yellow, fill opacity=0.5]
                   (1.5, 0) -- (-1, 2) -- (1, 2) -- cycle;       
            \node[below = 0.2cm] at (-0.5, 0) {$p$};
            \node[below = 0.2cm] at (1.5, 0) {$q$};
        \end{tikzpicture}
        \caption{Left Inclusive}
    \end{subfigure}   
       \begin{subfigure}[b]{0.22\textwidth}
        \begin{tikzpicture}
            \draw [ultra thick, draw=black, fill=yellow, fill opacity=0.5]    
                   (-0.5, 0) -- (-1, 2) -- (1, 2) -- cycle;
            \draw [ultra thick, draw=black, fill=yellow, fill opacity=0.5]
                   (1.5, 0) -- (0, 2) -- (1, 2) -- cycle;       
            \node[below = 0.2cm] at (-0.5, 0) {$p$};
            \node[below = 0.2cm] at (1.5, 0) {$q$};
        \end{tikzpicture}
        \caption{Right Inclusive}
    \end{subfigure}     
       \begin{subfigure}[b]{0.22\textwidth}
        \begin{tikzpicture}
            \draw [ultra thick, draw=black, fill=yellow, fill opacity=0.5]    
                   (-0.5, 0) -- (-1, 2) -- (1, 2) -- cycle;
            \draw [ultra thick, draw=black, fill=yellow, fill opacity=0.5]
                   (1.5, 0) -- (-0.5, 2) -- (0.5, 2) -- cycle;       
            \node[below = 0.2cm] at (-0.5, 0) {$p$};
            \node[below = 0.2cm] at (1.5, 0) {$q$};
        \end{tikzpicture}
        \caption{Margined-Inclusive}
    \end{subfigure}     
    \caption{Four types of intersecting intervals of two players}
    \label{fig:intersect_intervals}
\end{figure}

We define the class of inclusion-free convex graphs as follows. Given an ordering of the items $<_I$ that satisfies the adjacency property, for any pair of players if the neighbourhood of one player is completely included in the neighbourhood of the other, the smaller neighbourhood must either contain the leftmost or the rightmost item of the bigger neighbourhood with respect to the ordering $<_I$. This property is sometimes referred to as the \emph{enclosure} property in the literature.

\begin{definition}[Enclosure Property]
Let $H = (X, \, Y, \, E)$ be a bipartite graph. An ordering $<_X$ of $X$ in $H$ has the enclosure property if for each pair of vertices $y, y' \in Y$ for which $N(y) \subseteq N(y')$ in $X$, the vertices of $N(y') \, \backslash \, N(y)$ appear consecutively in $<_X$. We say that the graph $H$ as the enclosure property if such an ordering of the vertices of $H$ exists.
\end{definition}

As we mentioned earlier, every convex graph admits such an ordering and we can find it in linear time. Let $<_I$ be the ordering of items in set $I$ that satisfies the adjacency property. For every vertex $p \in P$ let $[\ell_p, \, r_p]$ be the interval of the items adjacent to $p$. Based on the ordering $<_I$ on $I$, we define the following ordering on $P$.

\begin{definition}[Lexicographical Ordering of Players]
For a given ordering $<_I$ of items $I$, an ordering of players is called lexicographical if and only if a player $p$ is ordered before another player $q$ whenever $\ell_p <_I \ell_q$, or $\ell_p = \ell_q$ and $r_p \leq_I r_q$ (breaking ties arbitrarily), in which $r_p \leq_I r_q$ means either $r_p = r_q$ or $r_p <_I r_q$.
\end{definition}

The adjacency property is in some contexts referred to as the \emph{Min ordering} property. When a graph admits both the adjacency and the enclosure property, it is said to respect the \emph{Min-Max ordering} property. In a given convex graph $H = (I, \, P, \, E)$ with the ordering of items $<_I$ that satisfies the adjacency property, let $<_P$ denote the lexicographical ordering of the players based on $<_I$. For arbitrary players $p, q \in P$ and arbitrary items $i, j, \in I$, if $p$ is adjacent to $i \in I$ and  $q$ is adjacent to $j$ with $p <_P q$ and $j <_I i$, then by the adjacency property $p$ is also adjacent to $j$, meaning that the neighbourhood of $p$ contains the ``minimum'' of the two items $i$ and $j$ with respect to the ordering $<_I$. The enclosure property dictates that $q$ also be adjacent to $i$, the maximum of the two items $i$ and $j$, hence the name ``Min-Max'' ordering (not to be confused with Min-Max allocation problem) property. Note that not all convex graphs satisfy the enclosure (or Min-Max ordering) property. In fact, the class of convex graphs that satisfy the enclosure property is a proper subclass of the convex class and is known to be equivalent to two famous graph classes, the bipartite permutation graphs and the proper interval bigraphs.


The following observation shows an interesting property of graphs that satisfy both the adjacency and enclosure properties.

\begin{observation}
\label{obs:double-interval}
Assume we are given an inclusion-free convex graph $H = (I, \, P, \, E)$ alongside an ordering of items $<_I$ that satisfies both the adjacency and the enclosure properties. If $I$ is ordered according to $<_I$, then the lexicographical ordering of $P$ based on $<_I$ also satisfies both properties. In other words, when $P$ is ordered lexicographically based on $<_I$, (i) the neighbourhood of every item $x \in I$ forms an interval in $P$, and furthermore, (ii) for every pair of items $x_1, x_2 \in I$, their respective neighbourhoods are either properly overlapping, left inclusive, or right inclusive.
\end{observation}

\begin{proof} We first show part \emph{(i)}. For the sake of contradiction, assume there is an item $x \in I$ such that the neighbourhood of $x$ in $P$ does not form an interval. This implies that a gap exists in the neighbourhood of $x$ in $P$, meaning that there are players $p' <_P p <_P p''$ such that $x$ in adjacent to $p'$ and $p''$ but not $p$. Here, $<_P$ denotes the lexicographical ordering of players. Since the neighbourhood of players form intervals in $I$, the neighbourhood of $p'$ can include $x$ in one of the three ways depicted in \figref{double_interval1}. The dotted edges represented the non-edges. We go through each case:

\begin{description}
\item[Case 1:] $x$ is in the middle of the interval of $p'$. As $p$ is not adjacent to $x$, we first consider the case in which $\ell_p <_I x$. We can conclude $r_p <_I x$ since the neighbourhood of $p$ is an interval in $I$. This means that either $\ell_p = \ell_{p'}$ which together with the fact that $r_p <_I r_p'$ (note that $p'$ is adjacent to $x$ and $p$ is not, so the right boundary of the neighbourhood of $p'$ must be to the right of that of $p$ in this case) means that $p$ should come before $p'$ in the lexicographical ordering contradicting the assumption that $p' <_P p$, or we have that $\ell_{p'} <_I \ell p$ and $r_p <_I r_p'$ which contradicts the assumption that the graph satisfies the enclosure property (the neighbourhood of $p$ falls completely inside that of $p'$). Therefore it follows that $x <_I \ell_p$. In this case, since $p <_P p''$, we either have that $\ell_p = \ell_{p''}$ or $\ell_p <_I \ell_{p''}$. In both this situations, $p''$ cannot be adjacent to $x$, implying that the assumption of Case 1 over the neighbourhood of players is false.

\item[Case 2:] $x$ is on the left boundary of the interval of $p'$. In this case, we certainly have that  $x <_I \ell_p$ as $p$ is lexicographically larger than $p'$ and not adjacent to $x$. As in Case 1, either $\ell_p = \ell_{p''}$ or $\ell_p <_I \ell_{p''}$ due to the lexicographical ordering. In both situations, $p''$ cannot be adjacent to $x$ which makes Case 2 a contradiction as well. 

\item[Case 3:] $x$ is on the right boundary of the interval of $p'$. Again, $x <_I \ell_p$, as assuming otherwise leads to a contradiction we mentioned in Case 1. Therefore, since $p <_P p''$, we get the two cases as before, $\ell_p = \ell_{p''}$ or $\ell_p <_I \ell_{p''}$, both leading to a contradiction.
\end{description}

\begin{figure}
    \centering
    \begin{subfigure}[b]{0.32\textwidth}
        \begin{tikzpicture}
            \draw [black, thick] (0, 0) -- (-1, 2);
            \draw [black, thick] (0, 0) -- (1, 2);
            \draw [black, thick] (2, 0) -- (0, 2);            
            \draw [black, thick, dashed] (1, 0) -- (0, 2);          
            \draw [draw = black, fill = white, thick] (0, 0) circle (0.2cm);
               \draw [draw = black, fill = white, thick] (1, 0) circle (0.2cm);
              \draw [draw = black, fill = white, thick] (2, 0) circle (0.2cm);
            \draw [draw = black, fill = white, thick] (0, 2) circle (0.2cm);
               \draw [draw = black, fill = white, thick] (-1, 2) circle (0.2cm);
               \draw [draw = black, fill = white, thick] (1, 2) circle (0.2cm);            
            \node[below = 0.2cm] at (0, 0) {$p'$};
            \node[below = 0.33cm] at (1, 0) {$p$};
            \node[below = 0.2cm] at (2, 0) {$p''$};          
            \node[above = 0.2cm] at (0, 2) {$x$};
        \end{tikzpicture}
        \caption{$x$ in the Middle}
    \end{subfigure}
       \begin{subfigure}[b]{0.32\textwidth}
        \begin{tikzpicture}
            \draw [black, thick] (0, 0) -- (0, 2);
            \draw [black, thick] (0, 0) -- (1, 2);
            \draw [black, thick] (2, 0) -- (0, 2);            
            \draw [black, thick, dashed] (1, 0) -- (0, 2);          
            \draw [black, thick, dashed] (0, 0) -- (-1, 2);              
            \draw [draw = black, fill = white, thick] (0, 0) circle (0.2cm);
               \draw [draw = black, fill = white, thick] (1, 0) circle (0.2cm);
              \draw [draw = black, fill = white, thick] (2, 0) circle (0.2cm);
            \draw [draw = black, fill = white, thick] (0, 2) circle (0.2cm);
               \draw [draw = black, fill = white, thick] (-1, 2) circle (0.2cm);            
               \draw [draw = black, fill = white, thick] (1, 2) circle (0.2cm);            
            \node[below = 0.2cm] at (0, 0) {$p'$};
            \node[below = 0.33cm] at (1, 0) {$p$};
            \node[below = 0.2cm] at (2, 0) {$p''$};          
            \node[above = 0.2cm] at (0, 2) {$x$};
        \end{tikzpicture}
        \caption{$x$ on the Left Boundary}
    \end{subfigure}   
       \begin{subfigure}[b]{0.32\textwidth}
        \begin{tikzpicture}
            \draw [black, thick] (0, 0) -- (-1, 2);
            \draw [black, thick] (0, 0) -- (0, 2);
            \draw [black, thick] (2, 0) -- (0, 2);            
            \draw [black, thick, dashed] (1, 0) -- (0, 2);          
            \draw [black, thick, dashed] (0, 0) -- (1, 2);              
            \draw [draw = black, fill = white, thick] (0, 0) circle (0.2cm);
               \draw [draw = black, fill = white, thick] (1, 0) circle (0.2cm);
              \draw [draw = black, fill = white, thick] (2, 0) circle (0.2cm);
            \draw [draw = black, fill = white, thick] (0, 2) circle (0.2cm);
               \draw [draw = black, fill = white, thick] (-1, 2) circle (0.2cm);
               \draw [draw = black, fill = white, thick] (1, 2) circle (0.2cm);                        
            \node[below = 0.2cm] at (0, 0) {$p'$};
            \node[below = 0.33cm] at (1, 0) {$p$};
            \node[below = 0.2cm] at (2, 0) {$p''$};          
            \node[above = 0.2cm] at (0, 2) {$x$};
        \end{tikzpicture}
        \caption{$x$ on the Right Boundary}
    \end{subfigure}     
    \caption{Different cases a gap may exist in the interval of an item $x$ in an inclusion-free convex graph}
    \label{fig:double_interval1}
\end{figure}
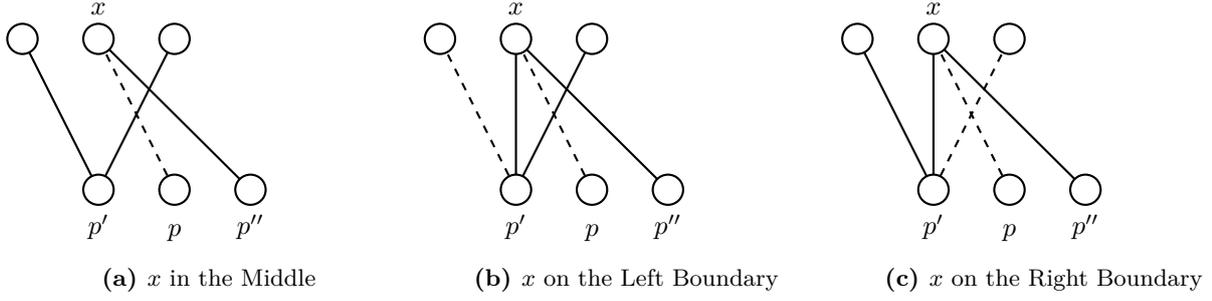

The assumption that a gap exists in the neighbourhood of $x$ leads to a contradiction in either case. This proves part \emph{(i)} of the observation. For part \emph{(ii)}, we assume there is an item $x$ whose neighbourhood is marginally included inside the neighbourhood of some other item $x'$. Again, there are two cases: one in which $x' <_I x$, and the other $x <_I x'$. We only show the former as the latter follows from symmetry. Since the interval of $x$ is entirely inside the neighbourhood of $x'$ without touching any of the boundaries, there must exist players $p' <_P p <_P p''$ such that $p$ is adjacent to both $x$ and $x'$, but $p'$ and $p''$ are only adjacent to $x'$, lying on the left and right side of the neighbourhood of $x$. \figref{double_interval2} depicts the neighbourhood of players and items, with the dashed edges representing the non-edges. Note that the neighbourhood of each of the players is an interval in $I$. Players $p'$ and $p''$ are both adjacent to $x'$ but not $x$, while player $p$ is adjacent to $x$. Since $x' <_I x$, we conclude that $\ell_p' <_I r_p' <_I r_p$ and $\ell_p'' <_I r_p'' <_I r_p$. This contradicts the assumption that $p''$ is lexicographically greater than $p$. Therefore the intervals of items in the set $P$ must be inclusion-free.
\end{proof}

Given a convex graph $H = (I, \, P, \, E)$, and ordering $<_I$ of items that satisfies the adjacency property, and a target value $t$, our goal is to find a $t$-assignment of items to players, defined informally as a feasible assignment such that each player (machine) receives total utility at most $t$. More precisely, 

\begin{definition}[$t$-assignment for the Max-Min case]
A $t$-assignment, for any $t \ge 0$, is a feasible assignment such that every player $p_j$ receives a set of items $I^j \subseteq [\ell_{p_j}, \, r_{p_j}]$ with total value \underline{at least} $t$.
\end{definition}

In this paper, we only deal with the easier case of inclusion-free instances. All the definitions and theorems in later sections apply to the class of inclusion-free convex graphs (unless mentioned otherwise explicitly). To the best of our knowledge, no PTAS is known for general instances with margined-inclusive intervals, where the convex graph only satisfies the adjacency property. Whether there exists a PTAS for convex graphs, in general, is an interesting research question that we pose as an open problem.

\subsubsection{A Note on the Notation}

Throughout the paper, we use \textbf{bold} mathematical symbols to refer to vectors. Superscripts are reserved for player indices. For simplicity, we denote the induced subgraphs of the input convex graph $H$ by $H' = (I', \, P')$ for some $I' \subseteq I$ and $P' \subseteq P$ and refrain from explicitly mentioning the subset of edges $E' \subseteq E$. The set $E'$ contains all those edges of $E$ that are incident to at least one vertex in $I'$ \underline{and} at least one vertex in $P'$. The \underline{underline} is used for emphasis. Finally, we use \emph{italics} to highlight the terms we define in the paper.

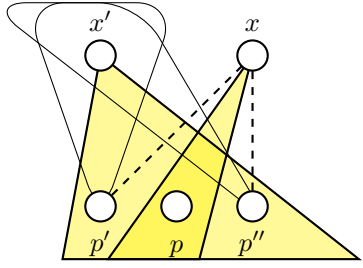
\begin{figure}
    \centering
    \begin{tikzpicture}
             \draw [fill = yellow, opacity = 0.4] (1, 2) -- (-0.9, -0.7) -- (0.3,  -0.7) -- cycle;
             \draw [fill = yellow, opacity = 0.4] (-1, 2) -- (-1.5, -0.7) -- (2.4, -0.7) -- cycle; 
             \draw [thick] (1, 2) -- (-0.9, -0.7) -- (0.3,  -0.7) -- cycle;
             \draw [thick] (-1, 2) -- (-1.5, -0.7) -- (2.4, -0.7) -- cycle;                        
             \draw [rounded corners = 15pt] (1.2, -0.2) -- (-2.5, 2.7) -- (-0.5,  2.7) -- cycle;
             \draw [rounded corners = 15pt] (-1, -0.2) -- (-2, 2.7) -- (0,  2.7) -- cycle;           
             \draw [dashed, thick] (-1, 0) -- (1, 2);
             \draw [dashed, thick] (1, 0) -- (1, 2);
            \draw [draw = black, fill = white, thick] (0, 0) circle (0.2cm);
               \draw [draw = black, fill = white, thick] (1, 0) circle (0.2cm);
              \draw [draw = black, fill = white, thick] (-1, 0) circle (0.2cm);
               \draw [draw = black, fill = white, thick] (-1, 2) circle (0.2cm);
               \draw [draw = black, fill = white, thick] (1, 2) circle (0.2cm);            
            \node[below = 0.33cm] at (0, 0) {$p$};
            \node[below = 0.2cm] at (1, 0) {$p''$};
            \node[below = 0.2cm] at (-1, 0) {$p'$};          
            \node[above = 0.2cm] at (-1, 2) {$x'$};
            \node[above = 0.2cm] at (1, 2) {$x$};
    \end{tikzpicture}
    \caption{The contradictory assumption that intervals of items in the set of players $P$ are not inclusion-free}
    \label{fig:double_interval2}
\end{figure}

\subsection{Preprocessing the Instance}

 Given a particular instance of the problem, we first simplify the input instance. For a positive rational number $t$, we may assume that the value of each item is at most $t$. If item $x_i$ has value $v_i >t$ then we set $v_i$ to $t$ without loss of generality. By a proper scaling, \emph{i.e.}, dividing each value by $t$, we may assume that the value of each item is in $[0, 1]$. Observe that a $t$-assignment becomes a $1$-assignment. We do a binary search to find the largest value of $t$ for which a $t$-assignment exists (in the Min-Max case we seek the smallest $t$ such that a $t$ assignment exists). The binary search is carried out in the interval $[0, \, \frac{1}{n}\sum_{i = 1}^{m} v_i]$ where 0 is an absolute lower bound, and $\frac{1}{n}\sum_{i = 1}^{m} v_i$ is an absolute upper bound of the optimal solution respectively (for the Min-Max case the binary search is carried out in the interval $[\frac{1}{n} \sum_{i = 1}^{m} v_i, \sum_{i = 1 }^{m} v_i]$).

For a subset $P' \subseteq P$ of players, let $N_{[H]}(P')$ be the union of the set of all neighbours of the players in $P'$ in the graph $H$. We let $N(P')$ denote this set whenever the graph $H$ is implied by the context. For an interval $[i,j]$ of  items, let $P[i,j]$ be the set of players whose \underline{entire} neighbourhood lies in $[i,j]$, that is $P[i,j] = \{ p \in P: N(p) = [\ell_p,r_p] \subseteq [i,j] \}$. For a subset $I' \subseteq I$ of items, let $val(I')$ denote the sum of values of all the items in $I'$. We note that in every $1$-assignment, for every subset $P' \subseteq P$ of players, the sum of the value of items in its neighbourhood should be at least $|P'|$. In other words, $\forall P' \subseteq P: \, val(N(P')) \geq |P'|$. If the value of each item is $1$, then this condition is the well known Hall's condition \cite{hall}, a condition sufficient and necessary for a bipartite graph to have a perfect matching. We consider a more general version of the $1$-assignment and derive a generalized version of Hall's condition below. Each player $p \in P$ has a demand $d(p)$. This version contains the $1$-assignment as a special case, i.e., the case where $d(p) = 1$ for all players $p \in P$. Also, for a subset of players $P' \subseteq P$, let $d(P')$ denote the sum of demands of all players in $P'$. Now the generalized Hall's condition (for the Max-Min case) becomes: $\forall P' \subseteq P: \, val( N(P') ) \geq d(P')$. From now on we refer to this condition as Hall's condition. Note that this condition is necessary to satisfy the players' demands, but not sufficient (see \figref{counter_hall}). \lemref{hall} shows that in order to check Hall's condition for $H$ it suffices to check it for every interval of items. Therefore, Hall's condition in our setting becomes \condref{hall} below:

\begin{equation}
\forall \,\, [\ell, r] \subseteq [1,\, m]:  \qquad  val([\ell,r]) \ge d(P[\ell,r]). \label{cond:hall}
\end{equation}

\begin{lemma}\label{lem:hall}
In order to check Hall's condition for $H$ it suffices to verify \condref{hall}. In other words, it suffices to check Hall's condition
for every set of players $P[\ell,r]$, $[\ell,r] \subseteq [1, \, m]$.
\end{lemma}

\begin{proof}
Assume we are given a convex graph $H$. We claim that it is sufficient to check Hall's condition only for intervals of items, meaning that if there are any violations of Hall's condition, then there is at least one violation over an interval of items. In other words, we show that there exists a subset of the players $P'$ for which \emph{(i)} $N(P')$ consists of several maximal intervals (therefore $N(P')$ is not an interval itself) and \emph{(ii)} Hall's condition is violated for $P'$ if and only if there exists an interval of items $[\ell, \, r] \subseteq I$ for which \condref{hall} is violated.

\noindent
$\Rightarrow$: Assume that there exist several maximal intervals $J_1, \, J_2, \, \ldots, \, J_t$ whose union gives $N(P')$. Since each player is adjacent to an interval of items, there exists a corresponding partition of $P'$ into subsets $P'_1, \, P'_2, \, \ldots, \, P'_t$, such that $N(P'_i) = J_i$. Since Hall's condition is violated, we have

\begin{displaymath}
val(J_1) + \, \ldots \, + val(J_t) = val(N(P')) < d(P') = d(P'_1) + \, \ldots \, +  d(P'_t )
\end{displaymath}

\noindent
Thus, there must exist an $i, 1 \leq i \leq t,$ for which $val(J_i) = val( N(P'_i) ) < d(P'_i )$. Now let $\ell$ and $r$ be the leftmost and rightmost items in $J_i$ respectively. Since $N(P'_i) = J_i = [\ell, \, r]$, then the neighbourhood of every player in $P'_i$ should fall entirely in the interval $[\ell, \, r]$. Therefore, $P'_i \subseteq P([\ell, \, r])$. Thus, we conclude $val([\ell, \, r]) < d(P'_i) \leq d(P([\ell, \, r]))$, which is a violation of \condref{hall} for the interval $[\ell, \, r]$.

\noindent
$\Leftarrow$: 
Assume that \condref{hall} is violated for an interval of items $[\ell, \, r] \subseteq [1, \, m]$, meaning that $val([\ell, \, r]) < d(P([\ell, \, r]))$. We have that $N(P([\ell, \, r])) \subseteq [\ell, \, r]$ since otherwise $P([\ell, \, r])$ contains a player whose entire neighbourhood does not fall in the set $[\ell, \, r]$, which contradicts the definition of $P([\ell, \, r])$. Let $P'$ be $P([\ell, \, r])$. We then have that $val(N(P')) \leq val([\ell, \, r]) < d(P')$, which means that Hall's condition is violated for the set $P'$. This completes the second direction of the claim.
\end{proof}

As a result of \lemref{hall}, it is sufficient to check Hall's condition for every interval of items. Since there are at most $m^2$ intervals, Hall's condition can be verified in time polynomial in the size of $H$.

\begin{figure}
\centering
\includegraphics[scale=0.85]{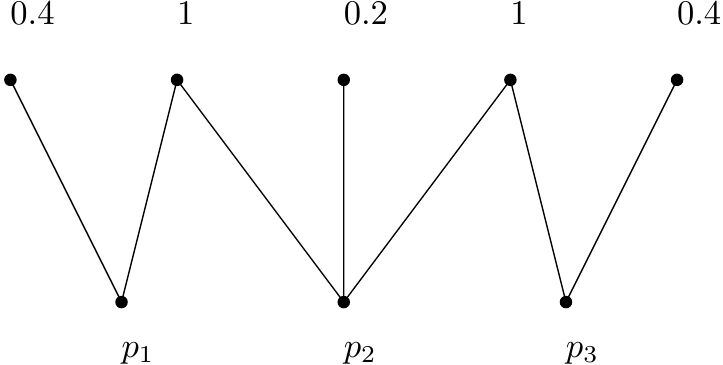}
\caption{An instance in which Hall's condition is satisfied for $t=1$ but the optimal solution value is not greater than $0.4$. In this examples, $d(p_1) = d(p_2) = d(p_3) = 1$.}
\label{fig:counter_hall}
\end{figure}

\subsubsection{Rounding the Instance}
\label{sec:rounding_santa}
In a given instance $H$ of the problem, we round the item values to obtain a rounded instance $H'$. For any integer $ k \geq 4$, we let $\frac{1}{k}$ be the error parameter. For each instance for which there is an optimal $1$-assignment, we seek an assignment such that each player receives a set of items with total value at least $1-\frac{4}{ k + 1 }$, $k \ge 4$.
\begin{itemize}
\item We call an item $x_i$ \emph{small} if its value is less than or equal to $\frac{1}{k}$. The values of small items are not changed in $H'$.
\item We call an item \emph{big} if its value is greater than $\frac{1}{k}$. If the value of item $x_i$, $v_i$, is in the interval $(\frac{1}{k}(1+\frac{1}{k})^\tau, \frac{1}{k}(1 + \frac{1}{k})^{\tau + 1}]$ for $\tau \geq 0$, it is rounded to $\frac{1}{k}(1 + \frac{1}{k})^{\tau + 1}$ in $H'$.
\end{itemize}

After the rounding, there are at most $C =  \lceil \log k / (\log ( 1+ \frac{1}{k}) ) \rceil$ \emph{categories} of big items, that is items with distinct values more than $\frac{1}{k}$. One can easily verify that $C$ is no more than $k^{1.4}$. For $0 \le \tau \le C$, let $q_{\tau} = \frac{1}{k}(1+\frac{1}{k})^{\tau}$. For subset $I'$ of $I$ let $v_s(I')$ denote the value of the small items in $I'$. The rounding process in the Min-Max case is slightly different and we include it in the corresponding section. 

\subsection{The Algorithm for Inclusion-Free Convex Graphs}
\label{sec:algorithm}

In this section, we consider the cases in which the neighbourhoods of no two players form a pair of margined-inclusive intervals. Note that left and right inclusion may still occur in such instances. The main result of this section is the following theorem.

\begin{theorem}\label{thm:finalItem}
Let $H$ be an instance of the problem before rounding with $n$ players and $m$ items. Then for $k \ge 4$ there exists a $(1-\frac{4}{k+1})$-approximation algorithm for the Max-Min allocation problem on inclusion-free convex graphs with running time $\bigO{(m + n)nm^{2(C+1)}}$, in which $C \leq k^{1.4}$.
\end{theorem}

Our proof technique builds upon a previous work of Alon et al. \cite{AAWY98}. In particular, we use the notion of ``input vectors'' (\defref{input_vec}) in a dynamic programming algorithm that considers the players one by one and remembers plausible allocations for them that can lead to a ($1 - \epsilon$)-assignment. Input vectors are configuration vectors that indicate the number of available items in an instance of the problem or one of its subproblems. In their paper, Alon et al. study complete bipartite graphs, \emph{i.e.}, the case where all items are available to all players. This fact allows them to ignore the actual assignment of items to players in their dynamic programming algorithm and only work with the configuration vectors for the set of available items, \emph{i.e.}, the input vectors; once the configuration vector for an assignment of items to a certain player is known, any set of items that matches the configuration vector can be assigned to the player. In other words, all items of a certain size are identical to the algorithm. Since the number of all possible input vectors is polynomial in the instance size, they manage to provide a polynomial-time algorithm in the end. 

Our work focuses on a more general case, that of inclusion-free convex graphs, which contains complete bipartite graphs as a subset. In this setting, an item may be adjacent only to a subset of players. As a result, to identify a subproblem, we cannot simply represent the set of available items to each player with the corresponding input vector. On the other hand, one cannot consider all possible ways of assigning items to players since the number of such assignments is exponential and computationally prohibitive. For the case of margined inclusion- free instances, we prove one can still identify the subproblems from their respective input vectors. To that end, we introduce ($1 - \epsilon$)-assignments of a certain structure (which we call simple-structured assignments) and prove their existence in any instance of the problem which admits a $1$-assignment. Such assignments are defined in \defrefs{alignment}{non-wasteful}, and shown to exist in \lemrefs{alignment}{non-wasteful_assign}.

In this section, we assume that the convex graph $H$ is a rounded input. For a rounded instance, it suffices to set the constant parameter $\epsilon$ to $\frac{3}{k}$. In the proof of \thmref{finalItem}, we show how this bound guarantees a $1 - \frac{4}{k + 1}$ approximation factor for the original instance. The dynamic programming algorithm has a forward phase and a backward phase. In the forward phase, the algorithm checks the existence of a solution by filling up a table. During this phase, the algorithm runs $n$ \emph{steps}, one for each player. At a given step $j$, it checks for feasible ($1 - \epsilon$)-assignments $\bcA' = (I^{n - j + 1}, \, I^{n - j + 2}, \, \ldots, \, I^n)$ for players $p_{n - j + 1}, \, p_{n - j + 2}, \, \ldots, \, p_n$, called \emph{partial} assignments, based on the solutions of the sub-problems in the previous step (feasible ($1 - \epsilon$)-assignments for players $p_{n - j + 2}, \, \ldots, p_n$). During the first step ($j = 1$), the algorithm searches for ($1 - \epsilon$)-assignments for player $p_n$ from scratch, therefore this case is treated differently. The backward phase also has $n$ steps. During this phase, the algorithm considers the players in the reverse order of the forward phase, \emph{i.e.}, from $p_1$ to $p_n$, and generates the final set of items to be allocated to each player based on the information stored in the table.

Every partial assignment of the forward phase indicates a \underline{potential} assignment of items to players\footnote{The assignment will be finalized once the algorithm finishes the steps of the forward phase.}. After every such assignment, the set of available items changes, thus a subproblem is instantiated. The induced subgraph of $H$ on the set of available items and unsatisfied players is called a \emph{remainder} graph. The algorithm regards every remainder graph as an instance of a subproblem of the original allocation problem. Note that the number of possible partial ($1 - \epsilon$)-assignments and their corresponding remainder graphs can be exponential. Some remainder graphs may eventually result in a solution and some may not. Thus, at each step, a na\"{i}ve dynamic programming approach might need to remember the remainder graphs, which will lead to an exponential memory (and time) complexity. By exploiting the properties of convex graphs (specifically the Min-Max ordering) we can show that one can do better than the na\"{i}ve approach. In this section, we examine some structural properties of the problem that allow us to find a solution for any given instance while checking a polynomial\footnote{In the number of players and items, but exponential in the inverse of the error parameter.} number of simple-structured assignments. These properties will also be useful for a more general case of the problem studied in \secref{margined}. Our dynamic programming algorithm uses input vectors, which we introduce here.

\begin{definition}[Input Vector]
\label{def:input_vec}
For a given convex graph $H$ for an input instance with $\nu_\tau$ big items of size $q_\tau$ for $\tau = 1, \, 2, \, \ldots, \, C$ and small items of total value in the interval $(\frac{\nu_0 - 1}{k}, \, \frac{\nu_0}{k}]$, an input vector is a configuration vector of the form $\nu(H) = \Nu = (\nu_0, \nu_1, \ldots, \nu_{C})$. 
\end{definition}

Note that for an arbitrary input vector $\Nu = (\nu_0, \nu_1, \ldots, \nu_{C})$, $\nu_\tau \leq m$ for all big items of size $q_\tau$, $1 \leq \tau \leq C$ as there can be at most $m$ items of a certain value. Furthermore, $\nu_0$ is the total value of the small items in integral multiples of $\frac{1}{k}$ which has been rounded up. Therefore, $\nu_0 \leq m$ since $\frac{m}{k}$ is an upper bound on the total value of the small items. As a result, there are at most $m^{C + 1}$ possible input vector values, which is a polynomial in the size of the problem instance. In the rest of this paper, we let $\cV$ denote the set of all possible input vectors. Similar to the algorithm by Alon et al., we are interested in a dynamic programming approach that only deals with the input vectors in $\cV$ rather than the (potentially exponential) remainder graphs. Unlike the case of complete graphs, we cannot assign items of a certain size interchangeably. However, by imposing the right set of restrictions on the assignments we can retrieve (reconstruct) a remainder graph from any input vector whose total sum of items is almost equal to the sum of values indicated by the input vector. In particular, we wish to show the following two facts: 

\begin{fact}
\label{fact:existence}
Whenever there exists an arbitrary $1$-assignment of items to players for an instance of the problem, there also exists a restricted ($1 - \epsilon$)-assignment.
\end{fact}

\begin{fact}
\label{fact:retrieve}
Given that a $1$-assignment exists for an instance of the problem, then there also exists a polynomial-time algorithm that, given the input vectors for every step (partial assignment) of a ($1 - \epsilon$)-assignment, reconstructs the remainder graph of every step, only from its input vector and not knowing the actual partial assignment, in such a way that the total value of the items in each reconstructed remainder graph is 
\begin{itemize}
\item the same as that of the corresponding remainder graph in the original (and unknown) $1$-assignment for every category $\tau = 1, \, 2, \, \ldots, \, C$ of the big items and,
\item only a small fraction $\epsilon$ more than its counterpart in the original remainder graph for the small items.
\end{itemize}
\end{fact}

As the first of the restrictions we impose on the assignment to obtain remainder graphs of simpler structures, we define \emph{right-aligned} assignments.

\begin{definition}[Right-Aligned Assignment]
\label{def:alignment}
For a given convex graph $H$, an ordering of items $<_I$ that satisfies the adjacency property, and a lexicographical ordering of players with respect to $<_I$, an assignment of items to players $\bcA = (I^1, \, I^2, \, \ldots, \, I^n)$ is called right-aligned if and only if it satisfies the following properties recursively:
\begin{enumerate}
\item Let $p_n$ be the last (rightmost) player in the graph. All the big items of value $q_{\tau} = \frac{1}{k}( 1 + \frac{1}{k} )^{\tau}$ ($1 \leq \tau \leq C$) assigned to $p_n$ are the rightmost ones with that value in $N_{[H]}(p_n)$. Furthermore, all the small items assigned to $p_n$ are also the rightmost ones in her neighbourhood.
\item In the remainder graph $H' = (I \, \backslash \, I^n, \, P \, \backslash \, p_n)$, that is if $p_n$ and all the items assigned to her are removed from the graph $H$, the rest of the assignment, $\bcA' = (I^1, \, I^2, \, \ldots, I^{n - 1})$, is a right-aligned assignment.
\item An empty assignment  (i.e., when no players are left in the graph) is considered a right-aligned assignment.
\end{enumerate}
\end{definition}

Next, we prove \factref{existence} for the right-aligned assignments in a lemma called the \emph{Alignment Lemma} (\lemref{alignment}). Intuitively, the lemma states that if a 1-assignment exists, then there exists a ($1 - \epsilon$)-assignment with items aligned to the right, in which $\epsilon = \frac{1}{k}$. We first need the following observation and definition.

\begin{observation}
\label{obs:hereditary}
Convex graphs have the hereditary property, meaning that every induced subgraph of a convex graph is also convex. 
\end{observation}

\begin{proof}
Assume otherwise, \emph{i.e.}, there exists no ordering of items that satisfies the adjacency property for a subgraph $H' = (I', \, P')$ of a convex graph $H = (I, \, P, \, E)$, where $P' \subseteq P$ and $I' \subseteq I$. Since $H$ is convex, there exists an ordering $<_I$ for $I$ which satisfies the adjacency property. On the other hand, $<_I$, when induced on the subset of items $I'$ does not satisfy the adjacency property in $H'$. Let $<_{I'}$ denote this induced ordering. There must exist a player $p \in P'$ and items $x, y, z \in I'$ for which $x$ and $z$ are in the neighbourhood of $p$ in $H'$ and $y$ is not while $x <_{I'} y <_{I'} z$. If we consider these items in $H$, $x$ and $z$ are still in the neighbourhood of $p$ and $y$ is still not, while $x <_I y <_I z$. This contradicts the assumption that $<_I$ satisfies the adjacency property in $H$. 
\end{proof}

\begin{definition}[Assignment Vector]
\label{def:assign_vec}
For an assignment of items to players $\bcA = (I^1, \, I^2, \, \ldots, \, I^n)$, and for a player $p_j$, $1 \leq j < n$,  let $\Alpha^j = \nu(H^j)$ for $j = 1, 2, \ldots, n$, in which $H^j = (P \, \backslash \, \{ p_n, \, p_{n - 1}, \, \ldots, p_{j + 1} \}, \, I \, \backslash \, (I^n \cup I^{n - 1} \cup \ldots \cup I^{j + 1}))$ is the remainder graph for player $p_j$ in the assignment $\bcA$ (for $j = n$, let $H^n = H$). We call the vector $\Alpha^j = (\alpha_0^j, \, \alpha_1^j, \, \ldots, \, \alpha_C^j)$ the \emph{assignment vector} of player $p_j$ in the assignment $\bcA$. Furthermore, let $\Alpha = (\Alpha^1, \, \Alpha^2, \, \ldots, \, \Alpha^n)$ be the assignment vector for the entire assignment $\bcA$.
\end{definition}

\begin{lemma}[The Alignment Lemma]\label{lem:alignment}
Suppose there exists a $1$-assignment for a given convex graph $H$. Then, there also exists a \underline{minimal} $(1-\frac{1}{k})$-right-aligned assignment for $H$. Furthermore, the two assignments will have identical assignment vector.
\end{lemma}

\begin{proof}

We let $\bcA = (I^1, \ldots, I^n)$ denote the $1$-assignment stated in the lemma and let $\Alpha$ be its assignment vector. Also, we denote by $I_{\tau}^j$ the set of big items of value $q_{\tau}$ assigned to player $p_j$ for $1 \leq \tau \leq C$ and by $I_0^j$, the set of small items allocated to $p_j$ in the assignment $\bcA$. We prove the lemma by transforming the assignment $\bcA$ into a minimal $(1 - \frac{1}{k})$-right-aligned assignment, $\bcA' = (J^1, \, J^2, \, \ldots, \, J^n)$. $\bcA'$ is minimal in the sense that the removal of any items from the sets $J^1, \, \ldots, J^n$ will cause its value to drop below $1 - \frac{1}{k}$. This task is carried out in two independent and consecutive rounds. In round $1$, the big items assigned to each player are aligned to the right while the small items are aligned in round $2$. Each round consists of $n$ steps, one for each player. The two rounds are independent of each other in that the alignment of the big items in the first round has no impact on the alignment of small items in the second round. In the following, we assume the players are ordered lexicographically based on the ordering of items $<_I$. 

\textbf{Round $1$:} In round $1$, the big items are right-aligned in $n$ steps, where in each step, the big items of a single player are aligned in $C$ \emph{micro-steps}, one micro-step for each item category. We initialize $j$ to be $n$, indicating that we start with the last player, $p_{n}$, and let $\tau$ be $1$. Also, let $H^j$ denote the remainder graph at the beginning of step $j$. At the $\tau^{th}$ micro-step, we only look at items of size ${q_\tau}$ in the graph and replace the set  $I_{\tau}^j$ with its right-aligned counterpart $J_{\tau}^j$. Note that according to \obsref{hereditary}, the induced subgraph on items of a certain size, say $H_{\tau}^j$, forms a convex graph. We start with $J_\tau^j = I_\tau^j$. If $I_{\tau}^j$ is right-aligned, \emph{i.e.}, all items of value $q_{\tau}$ assigned to $p_{j}$ are the rightmost ones in $N_{H_{\tau}^j}( p_{j} )$, then we return $J_\tau^j$. If however, $I_{\tau}^j$ is not right-aligned, there must be two items $x_{r}$ and $x_{t}$ of the same value $q_{\tau}$, where $x_{r} <_I x_{t}$ and $x_{r}$ is assigned to $p_{j}$, but $x_{t}$ is not. If $x_{t}$ is not assigned to any other player, then we can simply assign it to $p_{j}$ instead of $x_{r}$. If it is assigned to some player $p_{\ell}$, $x_{r}$ must also be connected to $p_{\ell}$. Otherwise $p_{\ell}$ would come after $p_{j}$ in the lexicographical ordering of the players. Since the two items have the same value, we use the adjacency property to swap them in the assignment. Thus, $p_{j}$ gets $x_{t}$ and $p_{\ell}$ gets $x_{r}$ in the transformed assignment $J_{\tau}^j$. Now we have one less item out of alignment. We continue this process until there are no items in $J_\tau^j$ out of alignment. Then we proceed to the next micro-step by setting $\tau \leftarrow \tau + 1$. At the end of $C$ micro-steps, we let $J^j$ be $\bigcup_{\tau = 1}^{C} J_{\tau}^j$. Furthermore, we obtain the remainder graph $H^{j - 1} = (I \, \backslash \,J^j, \, P \, \backslash \, \{ p_{j} \})$. This graph will be the input to the next step of the alignment. We now proceed to the next step by updating $j \leftarrow j - 1$. Note that the assignment vector for $\bcA' = (J^1, \, J^2, \, \ldots, \, J^n)$, say $\Alpha'$, is identical to $\Alpha$ on the big item coordinates since for any arbitrary player $p$, we did not change the number of big items of any category assigned to $p$. As a result, we have not lost any solution quality. 

\medskip
\textbf{Round $2$:} In round $2$, we obtain a right-aligned assignment for small items. We first explain the procedure and then prove its correctness. Let $H_0$ denote the subgraph of $H$ that contains only small items. The neighbourhood of each player in $H_0$ is still an interval (by \obsref{hereditary}). For $j \in [n]$, let $1 - w_{j}$ be the total value of the big items assigned to player $p_{j}$ in the $1$-assignment. Therefore, $w_{j}$ denotes the deficit of player $p_{j}$ that should be satisfied by small items. If the demands of each player $p_{j}$ is set to $w_{j}$ for $j \in [n]$, then $H_0$ must satisfy Hall's condition with these demands. This is because a 1-assignment exists in the original graph for which Hall's condition is necessary, and that during the first round, the values of the big items assigned to the players did not change. Following the idea of round $1$, we align the items in $n$ steps. We start with the last player, so we let $j$ be $n$ and let $I_0^j$ denote the set of small items assigned to $p_{j}$. We also let $H_0^j$ denote the remainder graph at the beginning of step $j$. Note that $val(I_0^j) = w_{j}$. We replace this set with the set of right-aligned items, $J_0^j$, of value at least $w_j - \frac{1}{k}$. We first let $J_0^j$ be the empty set. Then, we pick small items from the neighbourhood of $p_j$ in $H_0^j$ and add them to $J_0^j$. We continue to do this as long as $val(J_0^j) \leq w_{j} - \frac{1}{k}$. When the algorithm stops, since the value of each small item is at most $\frac{1}{k}$, $w_{j} - \frac{1}{k} < val(J_0^j) \leq w_j$. By doing this,  we ensure that player $p_{j}$ gets a value strictly greater than $w_{j} - \frac{1}{k}$ in small items, and strictly greater than $1 - \frac{1}{k}$ in total. The remainder graph for the next step of round $2$, $H_0^{j - 1}$, is obtained by $H_0^{j - 1} = (I \, \backslash \, J_0^j, \, P \, \backslash \{ p_j \})$. Note that each such remainder graph is still a convex graph because we started with a convex graph $H_0$ and at each step, we removed items from right to left in the ordering $<_I$. Also, note that the small coordinate of the assignment vectors of $\bcA$ and $\bcA'$ ($\Alpha$ and $\Alpha'$ respectively) are identical too, which means $\Alpha = \Alpha'$. The reason is that we packed the small items in $\bcA'$ is such a way that they have the same total value as in $\bcA$ when counted in integral multiples of $\frac{1}{k}$ and rounded up. We then move to the next step by updating $j \leftarrow j - 1$.

To prove the correctness, we show that after each step $j$, $1 \leq j\leq n$, we can obtain a right-aligned assignment of small items the sum of whose values is strictly greater than $w_{j}-\frac{1}{k}$ for player $p_j$. Recall that we chose every $J_0^j$ ($1 \leq j \leq n$) to be the minimal right-aligned set of items assigned to player $p_j$ for which $w_j - \frac{1}{k} < val(J_0^j) \leq w_j$. For the sake of contradiction, assume that for some $j$, $1 \leq j \leq n$, we cannot provide a right-aligned assignment of small items $J_0^j$ for which $val(J_0^j) > w_{j} - \frac{1}{k}$. Let $t$ be the largest such index. This assumption implies that $val(N_{[H_0^t]}(p_{t})) \leq w_{t} - \frac{1}{k}$. We let $H_0^{t, n}$ denote the induced subgraph of $H_0$ on players $p_{t}, \, p_{t + 1}, \, \ldots, \, p_{n}$ and the small items in their neighbourhoods. We consider the partial right-aligned assignments of small items in $H_0^{t, n}$ represented by $J_0^t, \, J_0^{t + 1}, \, \ldots, \, J_0^n$, and claim that for a subset of the players $p_{t}, \, p_{t + 1}, \, \dots \, p_{n}$, Hall's condition is violated for small items in the original graph $H_0$. Since Hall's condition is necessary for any feasible assignment that fulfills the demands of players, this in turn implies that there is no $1$-assignment for the instance, which contradicts our earlier assumption. To prove this claim, we introduce the notion of a \emph{gap} in the assignment. With respect to an ordering of items $<_I$, a gap exists in a (partial) assignment if $x_{j}$ is not assigned to any player, but there exists another item $x_{i}$, such that $x_i < x_j$ in $<_I$ and $x_{i}$ is assigned to some player. The item $x_j$ is said to be in the gap with respect to the partial assignment. Based on this notion, we consider two cases: 

\begin{description}
\item[Case 1:] \emph{There exists a gap in the partial assignment}. Let $u$  be the smallest index in $N_{[H_0^{t, n}]}(p_t, \, p_{t + 1}, \,$ $\dots, \, p_{n})$ for which item $x_{u}$ is in the gap for the partial right-aligned assignment. By the choice of $u$, $x_{u - 1}$ must be assigned to a player in the set $\{p_{t}, \, p_{t + 1}, \, \dots, \, p_{n}\}$. Let $p_{v}$ be that player. Note that $x_{u} \notin N_{[H_0^{t, n}]}(p_{v})$ as otherwise $x_{u}$ would be assigned to $p_{v}$ in the right-aligned assignment. Item $x_{u-1}$ is thus the rightmost item assigned to player $p_v$. Let $x_{\ell}$ and $x_{r}$ be the left-most and the right-most items in $N_{[H_0^{t, n}]}(p_t, \, p_{t + 1}, \, \dots, \, p_{v})$ respectively (note that $r= u - 1$). Then,
\begin{align*}
val([\ell, r]) & = \sum_{\begin{array}{c} j \in [\ell, r]: \\ x_j \in H_0^{t, n} \end{array}} val(x_{j}) \\
             & = \sum_{j = t}^{v} val(J_0^j) \\
             & \leq val(J_0^t) + \left( \sum_{j = t + 1}^{v} w_{j} \right) \\
             & \leq w_t - \frac{1}{k} + \left( \sum_{j = t + 1}^{v} w_{j} \right) \\
             & < \sum_{j = t}^{v} d(j).
\end{align*}
where the first inequality is due to the fact that the total value assigned to each player $p_j$ in the right-aligned assignment is at most $w_j$, the second inequality is due to the choice of $p_t$, and last inequality holds since $d_j = w_j$ for all players $p_j$. Thus, Hall's condition is not satisfied for the set of players $p_{t}, \, p_{t + 1}, \, \dots, \, p_{v}$ in $\hat{H}$.

\item[Case 2:] \emph{There are no gaps in the partial assignment}. Every item in $N_{[H_0^{t, n}]}(p_{t}, \, p_{t + 1}, \, \dots, \, p_{n})$ has been assigned in the partial assignment. Thus $\bigcup_{j = t}^{n} J_0^j = N_{[H_0^{t, n}]}(p_{t}, \, p_{t + 1}, \, \dots, \, p_{n})$. Once again, let $x_{\ell}$ and $x_{r}$ be the left-most and the right-most items in $N_{[H_0^{t, n}]}(p_t, \, p_{t + 1}, \, \dots, \, p_{v})$ respectively. Similarly, 
\begin{align*}
val([\ell, r]) & = \sum_{j \in [\ell, r]} val(x_{j}) \\
             & \leq v(J_0^t) + \left( \sum_{j = t + 1}^{n} w_j \right) \\
             & \leq \left(\sum_{j = t}^{n} w_{j}\right) - \frac{1}{k}  \\
             & < \sum_{j = t}^{n} d(j).
\end{align*}
This implies that Hall's condition is violated in $H_0$ for the set of players $p_{t}, \, p_{t + 1}, \, \dots, \, p_{n}$.
\end{description}
\end{proof}

Unfortunately, \factref{retrieve} is not true for the right-aligned assignment restriction, as shown in the following example.

\begin{exm}
\label{exm:private}
In the example shown in \figref{private}, each of the circle items is considered small items and have a value of $\frac{1}{10}$ and each of the square items, which are called big items, have a value of $\frac{1}{4}$. Now, consider two different right-aligned $1$-assignments. 
\begin{enumerate}
\item In the first one, player $p_3$ gets the 5 rightmost circle items in her neighbourhood as well as the 2 rightmost square items, player $p_{2}$ also gets the 5 rightmost circle items in her neighbourhood alongside the 2 rightmost square items, and player $p_{1}$ observes a remainder graph with input vector $\Nu = (5, 2)$. Fortunately for $p_1$, all these items are in her neighbourhood so that she can be assigned a $1$-assignment as well and all the players' demands are met.
\item In the second assignment, player $p_3$ gets all four square items in her neighbourhood. Player $p_2$ has only one choice if she is to be allocated a $1$-assignment and that choice is to get all the small circle items in her neighbourhood. At the end, player $p_1$ is left with a different remainder graph whose input vector is also $\Nu = (5, 2)$. This time, the five circle items remaining are not in the neighbourhood of $p_1$ (the five rightmost circle items), so $p_1$ has to get by with 2 square items which amount to only a $\frac{1}{2}$-assignment for her.
\end{enumerate}
\end{exm}

\begin{figure}
    \centering
    \begin{tikzpicture}
           \centering
        
        \draw [draw = black, fill = black, thick] (-2.5, 4.7) circle (0.1cm);
        \draw [draw = black, fill = black, thick] (-2.0, 4.7) circle (0.1cm);
        \draw [draw = black, fill = black, thick] (-1.5, 4.7) circle (0.1cm);
        \draw [draw = black, fill = black, thick] (-1.0, 4.7) circle (0.1cm);
        \draw [draw = black, fill = black, thick] (-0.5, 4.7) circle (0.1cm);
        \draw [draw = black, fill = black, thick] (0.5, 4.7) circle (0.1cm);
           \draw [draw = black, fill = black, thick] (1.0, 4.7) circle (0.1cm);
        \draw [draw = black, fill = black, thick] (1.5, 4.7) circle (0.1cm);
        \draw [draw = black, fill = black, thick] (2.0, 4.7) circle (0.1cm);
        \draw [draw = black, fill = black, thick] (2.5, 4.7) circle (0.1cm);
        \draw [draw = black, fill = black, thick] (3.5, 4.7) circle (0.1cm);    
        \draw [draw = black, fill = black, thick] (4.0, 4.7) circle (0.1cm);    
        \draw [draw = black, fill = black, thick] (4.5, 4.7) circle (0.1cm);
        \draw [draw = black, fill = black, thick] (5.0, 4.7) circle (0.1cm);
        \draw [draw = black, fill = black, thick] (5.5, 4.7) circle (0.1cm);
        
        \draw [draw = black, fill = black] 
            (-4.5, 4.4) -- (-4.0, 4.4) -- (-4.0, 3.9) -- (-4.5, 3.9) -- cycle;
        \draw [draw = black, fill = black] 
            (-3.8, 4.4) -- (-3.3, 4.4) -- (-3.3, 3.9) -- (-3.8, 3.9) -- cycle;
        \draw [draw = black, fill = black] 
            (0.9, 4.4) -- (1.4, 4.4) -- (1.4, 3.9) -- (0.9, 3.9) -- cycle;
        \draw [draw = black, fill = black] 
            (1.6, 4.4) -- (2.1, 4.4) -- (2.1, 3.9) -- (1.6, 3.9) -- cycle;
        \draw [draw = black, fill = black] 
            (3.0, 4.4) -- (3.5, 4.4) -- (3.5, 3.9) -- (3.0, 3.9) -- cycle;
        \draw [draw = black, fill = black] 
            (3.7, 4.4) -- (4.2, 4.4) -- (4.2, 3.9) -- (3.7, 3.9) -- cycle;
                
        \node[below = 0.2cm] at (-3.0, 1.0) {$p_1$};
        \node[below = 0.2cm] at (0.0, 1.0) {$p_2$};
        \node[below = 0.2cm] at (3.0, 1.0) {$p_3$};
        \draw [ultra thick, draw = black, fill = yellow, fill opacity = 0.2]    
            (0.0, 1.0) -- (-3.0, 5.0) -- (3.0, 5.0) -- cycle;                 
        \draw [ultra thick, draw = black, fill = yellow, fill opacity = 0.2]    
            (3.0, 1.0) -- (6, 5.0) -- (0.0, 5.0) -- cycle;    
        \draw [ultra thick, draw = black, fill = yellow, fill opacity = 0.2] 
            (-3.0, 1.0) -- (-6.0, 5.0) -- (0.0, 5.0) -- cycle;            
            
        \draw [draw = black, fill = white, ultra thick] (0, 1.0) circle (0.2cm);
        \draw [draw = black, fill = white, ultra thick] (3.0, 1.0) circle (0.2cm);
        \draw [draw = black, fill = white, ultra thick] (-3.0, 1.0) circle (0.2cm);        
            
    \end{tikzpicture}
    \caption{Private items can introduce challenges for the dynamic programming scheme. In the figure, circle items are small, with a value of $\frac{1}{10}$, and square items are big, and have a value of $\frac{1}{4}$.}
    \label{fig:private}
\end{figure}
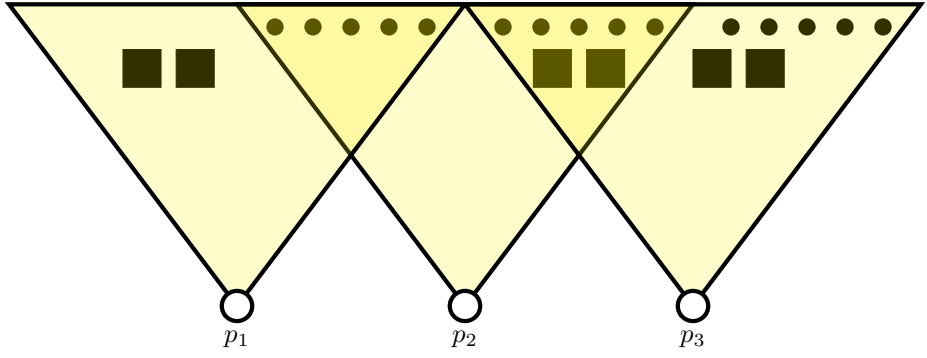

As this example shows, although we have restricted the assignments to right-aligned ones, we still cannot uniquely retrieve the proper (\emph{i.e.}, the one that provides a $1$-assignment for all the players) remainder graph from the input vector. Therefore, to avoid storing exponentially many graphs in the table $M$, we introduce more restrictions on the assignment. These added restrictions ensure that both facts would hold for the final assignment.

Despite the fact that the right-aligned restriction cannot establish a one-to-one mapping between the input vectors and remainder graphs, it takes us one step closer to such mappings. As illustrated by \exmref{private}, the reason we could not find an injective relation between the remainder graphs and the input vectors is that in some of the remainder graphs, there exist items that are not assigned to any player but are also not accessible to any player who has not yet received a bundle of items. For instance, for the second assignment in \exmref{private}, the five rightmost small items appear in the remainder graph, and consequently in the corresponding input vector, but are not in the neighbourhood of $p_1$. We call such items \emph{private}. More formally, we define private items in the following way.

\begin{definition}[Private Items]
For a subgraph $H' = (I', \, P')$ of the problem instance $H$, an item $x \in I'$ is called private if it has a degree of 1, that is, it is in the neighbourhood of only one player in $H'$. For a private item $x$ in $H'$, the player $p$ who is adjacent to $x$ is called the owner of $x$ in $H'$.
\end{definition}

\begin{definition}[Stranded Items]
For a subgraph $H' = (I', \, P')$ of the problem instance $H$, an item $x \in H'$ is called stranded if it has a degree of 0, that is, it is in the neighbourhood of no player in $H'$. 
\end{definition}

The private items, if kept unassigned in the graph, can mislead the retrieval procedure in that they appear as available items on the input vector, but become stranded items (their degree becomes zero once the player adjacent to them is removed from the graph, in which case they are wasted). We now define the new restricted assignments. 

\begin{definition}[Non-wasteful Right-Aligned Assignment]
\label{def:non-wasteful}
A right-aligned assignment $\bcA = (I^1, \, I^2, \, \ldots, \, I^n)$ is called a non-wasteful right-aligned assignment if none of the remainder graphs $H^j = (I \, \backslash \, (\bigcup_{t = j + 1}^{n} I^t), \, \{ p_1, \, p_2, \, \ldots,$ $\, p_j \})$ contains a stranded item for $j = 1, \, 2, \, \ldots, \, n - 1$.
\end{definition}


\noindent
We first mention the following observation. The proof is straightforward, and thus we sketch it briefly.

\begin{observation}
\label{obs:non-wasteful_existence}
Whenever there exists a right-aligned ($1 - \frac{1}{k}$)-assignment for a given instance, there also exists a non-wasteful right-aligned ($1 - \frac{1}{k}$)-assignment.
\end{observation}

\begin{proof}
This is because if in any right-aligned assignment, we take the unassigned private items in the remainder graphs and allocate them to their respective owners, we can only increase the value received by each player. Therefore, \factref{existence} holds for non-wasteful right-aligned assignments as well. 
\end{proof}

\noindent
\lemref{non-wasteful_assign} proves \factref{retrieve} for this type of restricted assignments.

\begin{lemma}
\label{lem:non-wasteful_assign}
Assume there exists a non-wasteful right-aligned $1$-assignment $\bcA = (I^1, \, I^2, \, \ldots, \, I^n)$ for an instance of the problem. Also assume we are given the original inclusion-free convex graph $H$ and the input vectors for partial assignments $\bcA^j = (I^j, \, I^{j + 1}, \, \ldots, \, I^n)$ for $ j = n - 1, \, n - 2, \, \ldots, \, 1$, but not the partial assignments themselves. There exists a polynomial time algorithm that can reconstruct the remainder graphs for every $\bcA^j$ in such a way that the total value of items in each reconstructed graph is at most $\epsilon = \frac{2}{k}$ more than the total value of the items in the actual corresponding remainder graph for $\bcA^j$.
%
%
%
\end{lemma}

\begin{proof}

We prove the lemma by providing the reconstruction algorithm. The algorithm is, in fact, a simple left-to-right sweep\footnote{Note that the items assigned are aligned to the right, and naturally, the remaining items in the graph would be aligned to the left.} of the items. We prove the correctness in \clmref{no_margin}. For $j = n, \, n - 1, \, \ldots, \, 1$, we generate a remainder graph after each partial assignment $\bcA^j = (I^j, \, I^{j + 1}, \, \ldots, \, I^n)$, using only the input vectors. Let $\Nu = (\nu_0, \, v_1, \, \ldots, \, \nu_C)$ be the input vector of the remainder graph after this partial assignment. Let $\hat{H} = (\hat{I}, \, \hat{P})$ denote the remainder graph that our algorithm outputs at the end of this step. We set $\hat{P}$ to be $\{p_1, \, p_2, \, \ldots, \, p_j\}$. As for the set of items $\hat{I}$, for the big items of value $q_{\tau}$, $\tau = 1, \, \ldots, \, C$, we select the left-most $\nu_{\tau}$ items in the graph $H$ and add them to $\hat{I}$. For small items, we select a \underline{maximal} set of left-most items whose total value is less than $\frac{\nu_0 + 1}{k}$. This method of selection ensures that the items that have been assigned to the players thus far are right-aligned. Furthermore, if at any point, stranded items (of degree $0$) appear in the remainder graph, we will return \textsf{NULL}. This ensures that private items are not left unassigned. Note that since $\bcA$ is assumed to be a non-wasteful right-aligned assignment, so are all its partial assignments as well. Next, we show the correctness of the algorithm by proving \clmref{no_margin}.

\begin{claim}
\label{clm:no_margin}
In a given inclusion-free instance $H$, for a non-wasteful right-aligned partial assignment $\bcA^j = (I^j, \, I^{j + 1}, \, \ldots, \, I^n)$, the set of big items in $\hat{H}$, the remainder graph reconstructed by the right-to-left sweep algorithm, is identical to that of the original remainder graph of $\bcA^j$, and the set of small items of $\hat{H}$ is a superset of the set of small items of the original remainder graph. Furthermore, the sum of values for small items in $\hat{H}$ is at most $\frac{2}{k}$ more than that of the original remainder graph.
\end{claim}

\begin{proof}
We treat big items and small items separately.

\smallskip
\noindent
\textbf{Big items:} First, we consider the induced graph on the set of big items of a certain size $q_{\tau}$ for some arbitrary $\tau = 1, \, \ldots, \, C$. Assume that $\hat{P} = \{ p_1, \, p_2, \, \ldots, \, p_j \}$ for some $1 \leq j < n$ (the case where $j = n$ is trivial since the remainder graph would be the original graph $H$). Also let $H'$ denote the original remainder graph for $\bcA^j$. For contradiction, assume that the graphs $H'$ and $\hat{H}$ are not identical when induced on the set of big items of value $q_{\tau}$. Since both graph $H'$ and $\hat{H}$ have the same number of items of value $q_{\tau}$, there must exist items $x'$ and $\hat{x}$ such that $x'$ appears in $H'$, but not in $\hat{H}$ and $\hat{x}$ belongs to $\hat{H}$, but not to $H'$. This means that the item $\hat{x}$ is assigned to a player $p_x$, $x \geq j + 1$ in one the earlier partial assignments $\bcA^{j  + 1}, \, \bcA^{j + 2}, \, \ldots, \, \bcA^n$. The item $x'$ is either adjacent to a player $p_y$, $1 \leq y \leq j$, or it is an item of degree $0$ in $H'$ and $\hat{H}$ (note that players $p_{j + 1}, \, \ldots, \, p_n$ do not belong to the graphs $H'$ and $\hat{H}$). If the latter case occurs, it means that $x'$ is a stranded item in $H'$, contradicting the fact that $H'$ is the remainder graph of a non-wasteful right-aligned assignment. In fact, when faced with a stranded item, the reconstruction algorithm returns \textsf{NULL} which signals the dynamic programming algorithm to skip the current input vector as it does not correspond to a valid remainder graph. Therefore, we can assume there exists a player $p_y$ adjacent to $x'$. Since $y \leq j$ and $x \geq j + 1$, we have that $p_y < p_x$. By the adjacency property, we have that $\hat{x} \in N_{[H]}(p_y)$. $p_x$ is not adjacent to $x'$ since otherwise the item $\hat{x}$ assigned to her would be out of alignment. The item $x'$ is assigned to the player $p_y$, therefore we conclude that $r_{p_x} < r_{p_y}$. Since $p_y < p_x$, it must be the case that $\ell_{p_y} < \ell_{p_x}$. This is a contradiction to the assumption that there are no inclusions in $H$. This situation is depicted in \figref{no_margin}.

\smallskip
\noindent
\textbf{Small items:} Assume that the set of small items of $\hat{H}$ is not a superset of the set of small items of $H'$, meaning that an item $x'$ exists such that $x'$ belongs to $H'$ but not to $\hat{H}$. If there exists another item $\hat{x}$ such that $\hat{x}$ belongs to $\hat{H}$, but not to $H'$, then, similar to the situation in \figref{no_margin}, we conclude that a margined inclusion should exist in $H$, which is a contradiction. Therefore, we may assume that $\hat{H}$ is a subgraph of $H'$. Recall that $\Nu = \nu(H') = (\nu_0, \, \nu_1, \, \ldots, \, \nu_C)$ and the sum of values of all the small items in $H'$ is at most $\frac{\nu_0}{k}$ (including $x'$). The total value of small items in $\hat{H}$ is at least $\frac{\nu_0}{k}$ by definition. This is a contradiction since we assumed $\hat{H}$ is a subgraph of $H'$. Therefore, the set of small items of $\hat{H}$ must be a superset of $H'$. In the proposed graph $\hat{H}$, we selected the minimal set of small items whose value is greater than or equal to $\frac{\nu_0}{k}$. Such a set may have a value as high as $\frac{\nu_0 + 1}{k} - \epsilon$ for some small value of $\epsilon > 0$ since the value of small items is less than or equal to $\frac{1}{k}$. The minimum total value of small items in $H'$ is no less than $\frac{\nu_0 - 1}{k}$ by the definition of input vectors, hence the difference of at most $\frac{2}{k}$. This complete the proof of the claim.
\end{proof}

\begin{figure}
    \centering
    \begin{tikzpicture}
           \draw[black, dashed, thick]        
               (0, 0) -- (2, 2);
           \draw[decoration={markings, mark=at position 0.3 with {\arrow{<}}}, postaction={decorate}, thick]        
               (2, 0) -- (0, 2);  
           \draw [red, thick, dashed] (0, 0) -- (0, 2);
        \draw [draw = black, fill = white, thick] (0, 2) circle (0.2cm);
           \draw [draw = black, fill = white, thick] (2, 2) circle (0.2cm);
        \node[below = 0.2cm] at (0, 0) {$p_y$};
        \node[below = 0.2cm] at (2, 0) {$p_x$};
        \node[above = 0.2cm] at (0, 2) {$\hat{x}$};
        \node[above = 0.2cm] at (2, 2) {$x'$};
        \draw [ultra thick, draw=black, fill=yellow, fill opacity=0.2]    
               (0, 0) -- (-3, 2.8) -- (4, 2.8) -- cycle;                 
        \draw [ultra thick, draw=black, fill=yellow, fill opacity=0.2]    
               (2, 0) -- (-2.3, 2.8) -- (1.1, 2.8) -- cycle;    
        \draw [draw = black, fill = white, thick] (0, 0) circle (0.2cm);
        \draw [draw = black, fill = white, thick] (2, 0) circle (0.2cm);
        \end{tikzpicture}
    \caption{If the set of big items of $H'$ and $\hat{H}$ are not identical, a margined inclusion should exist. The arrow indicates assignment in one of the predecessors of $H'$.}
    \label{fig:no_margin}
\end{figure}
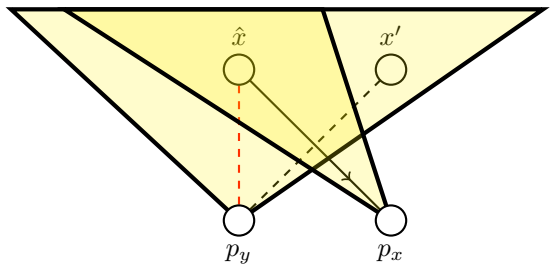

\end{proof}

\begin{observation}
\label{obs:approx_factor}
Note that, as mentioned in \clmref{no_margin}, the algorithm tends to overestimate the value of small items still available in the graph by at most an additive constant factor of $\frac{2}{k}$, but it never underestimates. Therefore when a player $p_j$ expects a bundle of small items worth $d_j$, she might receive $d_j - \frac{2}{k}$ worth of small items instead.
\end{observation}

Based on \obsref{approx_factor}, we next provide a dynamic programming algorithm that enumerates all non-wasteful right-aligned ($1 - \frac{3}{k}$)-assignments for a rounded instance $H$. From \lemref{alignment} and \obsref{non-wasteful_existence} one can conclude that whenever a $1$-assignment exists for a given instance of the problem, so does a non-wasteful right-aligned $(1 - \frac{1}{k})$-assignment. Therefore, if we only consider non-wasteful right-aligned assignments instead of all assignments, we are still able to find an assignment with objective value at least $(1 - \frac{1}{k})$, provided that a $1$-assignment for the instance exists in the first place. \lemref{non-wasteful_assign} and \obsref{approx_factor} together show that for any non-wasteful right-aligned assignment for an instance of the problem, an approximation of the assignment can be retrieved, and we only lose an additive fraction of $\frac{2}{k}$ in solution quality.  Therefore, the general idea of the algorithm is to consider all the $m^{C + 1}$ possible input vectors in a dynamic programming algorithm, retrieve  a non-wasteful right-aligned $(1 - \frac{3}{k}$)-assignment for each vector, and mark the feasible one. If it eventually manages to find at least one feasible assignment for $p_1$ (the last player considered by the algorithm), then it reports success. Before presenting the algorithm, we first introduce two functions that are invoked by the algorithm:

\begin{enumerate}
\item \textsf{retrieve}: Function \textsf{retrieve} takes as input parameters input vector $\Nu$ and player index $j$, and returns remainder graph $H' = (I', \, P')$. The graph $H' = (I', \, P')$ that it reconstructs is such that $P' = (p_1, \, p_2, \, \ldots, \, p_j)$ and $I'$ is formed by the left-to-right sweep algorithm in \lemref{non-wasteful_assign} \footnote{In fact, Function \textsf{retrieve} also needs the input inclusion-free convex graph $H$ as the third input parameter. To simplify notation, we omit this parameter from the function calls.}. Function \textsf{retrieve} returns \textsf{NULL} if at any point during reconstruction, there is at least one stranded item in the remainder graph.
\item \textsf{feasible}: Function \textsf{feasible} takes as input parameters remainder graph $H'$ and set of items $I'$, and returns either \textsf{true} or \textsf{false}. It returns \textsf{true} if  each of the following is true: $H'$ is not \textsf{NULL}, set $I'$ is entirely in the neighbourhood of the last player in $H'$, and the sum of item values in $I'$ is at least $1 - \frac{3}{k}$. Otherwise, Function \textsf{feasible} returns \textsf{false}.
\end{enumerate} 

\noindent
\textbf{Algorithm Assignment (for Inclusion-Free Convex Instances):} Algorithm Assignment uses dynamic programming to fill the entries of an $n \times m^{C + 1}$ table, in which each row represents a player, and each column represents an input vector (we assume an arbitrary ordering of all valid input vectors). Let the matrix $M$ denote this table.  Each entry $(j, \Nu)$ of $M$ is composed of two fields:
\begin{enumerate}
\item $M(j, \, \Nu).\textsf{bit}$ is a binary variable. If $M(j, \, \Nu).\textsf{bit}$ is set to $1$, then the input vector $\Nu$ is called \emph{marked} for player $p_j$. Otherwise, the vector $\Nu$ is \emph{unmarked} for player $p_j$.

\item $M(j, \, \Nu).\textsf{ptr}$ either contains \textsf{NULL} or a reference to another input vector $\Nu'$ (in the beginning,  this pointer is initialized to \textsf{NULL} for every entry).
 
\end{enumerate}

Let $\Nu_{in} = (\nu_{0}, \, \nu_{1}, \, \ldots, \, \nu_{C})$ be the input vector for the original graph $H$. The algorithm starts the forward phase with the last player $p_n$. For every input graph $\Nu$ (\emph{i.e.}, every column of the table $M$), it forms the remainder graph $\hat{H} = \textsf{retrieve}(\Nu, \, n - 1)$. Let $\hat{I}$ denote the set of items in $\hat{H}$ and also let $I^n = I \, \backslash \, \hat{I}$. The set $I^n$ provides a potential assignment for player $p_n$. The algorithm checks if this assignment is a non-wasteful right-aligned ($1 - \frac{3}{k}$)-assignment for $p_n$ using the boolean function $\textsf{feasible}(H, \, I^n)$. Note that for the first step of the algorithm, the original graph $H$ serves as the remainder graph. We call input vector $\Nu$ a \emph{successful} input vector for remainder graph $H'$ and player $p_n$ if $p_n$ is the last player in $H'$ and a call to the function $\textsf{feasible}(H', \, I^n)$ returns \textsf{true}. All the successful assignments are marked by setting $M(n, \, \Nu).\textsf{bit} = 1$. For a successful input vector $\Nu$ for player $p_n$, we set $M(n, \, \Nu).\textsf{ptr} = \Nu_{in}$, while for unsuccessful vectors, the value of $M(n, \, \Nu).\textsf{ptr}$ is not changed. In this section, we will not use the field \textsf{ptr} in the forward phase of the algorithm\footnote{We will however make use of this field in the forward phase of the algorithm in \secref{margined}}.

After this first iteration for player $p_n$, the algorithm proceeds to a similar procedure for other players. For a player $p_j$, $1 \leq j \leq n - 1$, it seeks ($1 - \frac{3}{k}$)-assignments once more. The only difference is that since the assignment should provide a ($1 - \frac{3}{k}$)-assignment for all players $p_n, \, p_{n - 1}, \, \ldots, \, p_j$, for every entry $(j, \, \Nu)$ the algorithm looks at all marked entries of the previous row, $j + 1$. Let $M(j + 1, \, \Nu')$ be such an entry. The algorithm retrieves the graphs $\hat{H} = \textsf{retrieve}(\Nu, j - 1)$ and $\hat{H}' = \textsf{retrieve}(\Nu', j)$. We let $\hat{I}$ and $\hat{I}'$ denote the item sets of graph $\hat{H}$ and $\hat{H}'$ respectively. If the set of items $I^j = \hat{I}' \, \backslash \, \hat{I}$ is confirmed to be at least a ($1 - \frac{3}{k}$)-assignment for player $p_j$ by function $\textsf{feasible}(\hat{H}', \, I^j)$, the entry $M(j, \, \Nu).\textsf{bit}$ is set to $1$. Furthermore, if $M(j, \, \Nu).\textsf{ptr}$ is \textsf{NULL}, it is updated to $\Nu'$ to mark this step of the assignment, indicating that the input vector $\Nu$ was achieved from a remainder graph represented by the input vector $\Nu'$. If the function \textsf{feasible} returns \textsf{false}, the algorithm simply moves to the next marked entry of the previous row. This procedure continues until the algorithm either finds a successful input vector for $p_j$, or that there is no successful vector in the previous row. In the former case, it reports success and moves to the backward phase in which the actual assignment is retrieved. In the latter case, it simply reports failure. In the backward phase, the non-wasteful right-aligned ($1 - \frac{3}{k}$)-assignment can be obtained by following the \textsf{ptr} back from the last row to the first. Before formally presenting the algorithm, we make the following remarks.

\begin{remark}
If Function \textsf{feasible} returns \textsf{false} at any stage, it may be due to any one of the following three reasons. Either \emph{(i)} the set $I^j$ does not have enough total item value, or \emph{(ii)} there are items in set $I^j$ which are not in the neighbourhood of player $p_j$, or \emph{(iii)} the remainder graph $\hat{H}$ returned by \textsf{retrieve} contains stranded items. In each of these cases, the input vector $\Nu$ does not represent a valid vector according to the set of restrictions imposed on the assignments. Therefore, the algorithm should ignore $\Nu$ and simply move to the next input vector.
\end{remark}

\begin{remark}
If the \textsf{ptr} field is already assigned a value, we do not change it since for the purpose of finding a feasible ($1 - \frac{1}{3}$)-assignments for margined-inclusive free instances, it is not important how we arrived at an input vector $\Nu$ in the current iteration. If a vector $\Nu$ represents an assignment that satisfies players $p_n, \, \ldots, \, p_{j + 1}$, a simple retrieval procedure that removes items from the right (as suggested in \lemref{non-wasteful_assign}) can retrieve a remainder graph that is a close enough approximation of the true remainder graph.
\end{remark}

\begin{remark}
In selecting the potential assignment $I^j$ for player $p_j$ at any iteration, we may assign a value more than the target value $1 - \frac{3}{k}$. Because we do not wish to leave items stranded or unaligned, the algorithm allocates every item in the difference of two consecutive remainder graphs to the player for whom the iteration is run. Nonetheless, as we formally discuss in \lemref{finalItem}, this is not at odds with our purpose of guaranteeing bundles of value $1 - \frac{3}{k}$ or higher in the rounded instance for every player.
\end{remark}

The pseudocode for the algorithm is provided below. In the following, a vector $\Nu$ is said to be less than or equal to $\Nu'$ if every entry of $\Nu$ is less than or equal to its counterpart in $\Nu'$.

\begin{algorithm}
 \KwData{rounded instance $H$.}
 \KwResult{returns either a ($1 - \frac{3}{k}$)-assignment for all players in the rounded instance, or report failure.}
 \For{every vector $\Nu$ in $\cV$ such that $\Nu \leq \Nu_{in}$} {
  $\hat{H} \leftarrow \textsf{retrieve}(\Nu, \, n - 1)$ \\
  $\hat{I} \leftarrow$ the set of items in $\hat{H} \, ; \, I^n \leftarrow I \, \backslash \, \hat{I}$ \\
  \If{$\textsf{feasible}(H, \, I^n)$} {
   $M(n, \, \Nu).\textsf{bit} \leftarrow 1$ \\ 
   \If{$M(n, \, \Nu).\textsf{ptr} == \textsf{NULL}$} {
    $M(n, \, \Nu).\textsf{ptr} \leftarrow \Nu_{in}$
   }
  }
 }
 \For{$j \leftarrow n - 1$ downto $1$}{
  \For{every vector $\Nu$ in $\cV$ such that $\Nu \leq \Nu_{in}$} {
   \For{every vector $\Nu'$ in $\cV$ such that $\Nu' \geq \Nu$ and $M(j + 1, \, \Nu').\textsf{bit} == 1$} {
    $\hat{H} \leftarrow \textsf{retrieve}(\Nu, \, j - 1) \, ; \, \hat{H}' \leftarrow \textsf{retrieve}(\Nu', \, j) $ \\
    $\hat{I} \leftarrow$ the set of items in $\hat{H}$ \, ; \, $\hat{I}' \leftarrow$ the set of items in $\hat{H}' \, ; \, I^j \leftarrow \hat{I}' \, \backslash \, \hat{I}$ \\
    \If{$\textsf{feasible}(\hat{H}', \, I^j)$} {
     $M(j, \, \Nu).\textsf{bit} \leftarrow 1$ \\ 
     \If{$M(j, \, \Nu).\textsf{ptr} == \textsf{NULL}$} {
      $M(j, \, \Nu).\textsf{ptr} \leftarrow \Nu'$
     }
    }
   }
  } 
 }
 \caption{Algorithm Assignment for Inclusion-Free Instances (Forward Phase)}
\end{algorithm}

In \lemref{finalItem} next, we show that Algorithm Assignment finds a ($1 - \frac{3}{K}$)-assignment in a rounded instance if one exists. We use \obsref{subset} in its proof. 

\begin{observation}
\label{obs:subset}
For a given instance $H$, assume $H'$ and $H''$ are two remainder graphs, and that $H'$ is an induced subgraph of $H''$. Let $\hat{H}' = \textsf{retrieve}(\nu(H'), \, j')$ and $\hat{H}'' = \textsf{retrieve}(\nu(H''), \, j'')$, where $j'$ and $j''$ are the indices of the last players in $H'$ and $H''$ respectively. Then, $\hat{H}'$ is also a subgraph of $\hat{H}''$.
\end{observation}

\begin{proof}
The function $\textsf{retrieve}$ preserves the set of players and the set of big items. Therefore, to prove the claim of the observation, it is enough to show that the set of small items in $\hat{H}'$ is a subset of the set of small items in $\hat{H}''$. Let $\Nu' = \nu(H')$ be $(\nu'_0, \, \nu'_1, \, \ldots, \nu'_C)$ and $\Nu'' = \nu(H'')$ be $(\nu''_0, \, \nu''_1, \, \ldots, \, \nu''_C)$. The fact that $H'$ is a subgraph of $H''$ implies that $\nu'_0 \leq \nu''_0$, which means that the function $\textsf{retrieve}$ has not packed more small items from the left to $\hat{H}'$ compared to $\hat{H}''$. Therefore, the set of small items of $\hat{H}'$ is a subset of that of $\hat{H}''$.
\end{proof}

\begin{lemma}\label{lem:finalItem}
Let $H$ be a rounded instance of the problem with $n$ players and $m$ items. The Algorithm Assignment assigns (in polynomial time) to each player a set of items with value at least $1 - \frac{3}{k}$ if a $1$-assignment exists for all players. The algorithm returns failure if no ($1 - \frac{1}{k}$)-assignment exists for $H$.
\end{lemma}

\begin{proof}

From \lemref{alignment}, we know that where there exists a $1$-assignment for a rounded instance $H$, there is also a right-aligned ($1 - \frac{1}{k}$)-assignment for the instance, which in turn implies the existence of a non-wasteful right-aligned ($1 - \frac{1}{k}$)-assignment. Now, assume the algorithm forms a bundle of items $I^j$ to be potentially assigned to an arbitrary player $p_j$. The way the algorithm forms the bundles is by finding the difference in the item sets of two remainder graphs, one before the assignment is made and one after. In doing so, it guesses the remainder graphs from their corresponding input vectors. Further, assume that the two remainder graphs before and after assignment are $H'$ and $H''$ respectively, thus $H''$ is a subgraph of $H'$. Also, the remainder graphs guessed by the function $\textsf{retrieve}$ are $\hat{H}'$ and $\hat{H}''$. Based on \obsref{subset}, $\hat{H}''$ is also a subgraph of $\hat{H}'$. In a worst case scenario, the function $\textsf{retrieve}$ may guess the total value of small items in $\hat{H}''$ to be at most $\frac{2}{k}$ more than the actual value (the value of small items in $H''$), while it guesses the same quantity correctly for $\hat{H}'$ (that is, equal to the total sum of small item values in $H'$). This means that the difference set $I^j$ would have $\frac{2}{k}$ less in small item values than the right-aligned ($1 - \frac{1}{k}$)-assignment implied by the original remainder graphs $H'$ and $H''$. Thus, the set $I^j$ results in a ($1 - \frac{3}{k}$)-assignment for player $p_j$. Therefore, for every possible non-wasteful right-aligned ($1 - \frac{1}{k})$-assignment to players, the algorithm considers a non-wasteful right-aligned ($1 - \frac{3}{k})$-assignment instead. Since the existence of the former results in the existence of the latter, Algorithm Assignment is guaranteed to find a ($1 - \frac{3}{k}$)-assignment to players if a ($1 - \frac{1}{k}$)-assignment exists.
On the other hand, when the algorithm reports failure, we can be certain that a ($1 - \frac{1}{k}$)-assignment does not exist, since otherwise, our algorithm would be able to find a slightly worse non-wasteful right-aligned ($1 - \frac{3}{k}$)-assignment that is guaranteed to exist.
\end{proof}

Now we can prove \thmref{finalItem}.

\begin{proof}[Proof of \thmref{finalItem}]
Each player receives a set of items with value at least $1-\frac{3}{k}$ for the rounded instance (see \lemref{finalItem}). Each item value is rounded up by at most a multiplicative factor of $1 + \frac{1}{k}$. Therefore each player receives a set of items with value at least $( 1 - \frac{3}{k} ) / ( 1 + \frac{1}{k} ) = 1 - \frac{4}{k + 1}$ of the optimal. The two inner \textbf{for} loops of Algorithm Assignment take $m^{2(C + 1)}$. The outer loop is run for $n$ players, and finding the remainder graphs and potential bundles need time $\bigO{m + n}$. Thus, the Algorithm Assignment runs in time $\bigO{(m + n)nm^{2(C + 1)}}$.
\end{proof}

\section{Min-Max Allocation Problem ($R \, | \, | \, C_{\max}$) On Convex Graphs}
\label{sec:rc_max}

In this section, we will show that an adaptation of the techniques used in the previous sections, will give us a PTAS for the Min-Max fair allocation of jobs to machines in a Min-Max allocation problem. Most of the techniques used in this section are very similar to those of \secref{max-min}, therefore we treat the proofs more concisely here, only highlighting the modifications and changes. 

\subsection{Problem Definition and Preliminaries}
\label{rc_max_intro}

As in instance of this problem, we are given a set $M = \{M_1, \, M_2, \, \ldots, \, M_m \}$ of identical machines or processors and a set $J = \{ J_1, \, J_2, \, \ldots, \, J_n \}$ of jobs. Each job $J_j$ has the same processing time $p_j$ on a subset of machines and it has processing time $\infty$ on the rest of the machines. The goal is to find an assignment of the (\underline{entire} set of) jobs to the machines $\boldsymbol{\cA} = (J^1, \, J^2, \, \ldots, \, J^m)$, such that the maximum load among all the machines is minimized. Formally, we are given a convex graph $H = (M, \, J, \, E)$ where $M$ is a set of machines and $J$ is a set of jobs and $E$ denotes the edge set. There is an edge in $E$ between a machine and a job if the job can be executed on that machine. Also given is a utility function $p: \, J \rightarrow \bQ_{>0}$ which assigns a processing time to each job. With slight abuse of notation, we write $p(J_j)$ as $p_j$ for short, and job $J_j$ requires $p_j$ units of processing time to run to completion on any machine $M_i$, provided that there exists an edge between $J_j$ and $M_i$ in the set $E$. If the edge does not exist, the processing time of $J_j$ on $M_i$ is not bounded. 

Similarly, as before, we assume that the input bipartite graph $H$ satisfies the adjacency and the enclosure properties, \emph{i.e.}, is an inclusion-free convex graph. In other words, we assume that we have an ordering $<_J$ of jobs (which orders the jobs as $J_1, \, J_2,...,J_n$ ) such that each machine can execute consecutive jobs (an interval of jobs). More formally, we denote the interval of job $J_i$ by $[\ell_i,r_i]$. We assume that $J_i$ is before $J_j$ and write $J_i <_J J_j$ whenever $\ell_i < \ell_j$ or $\ell_i = \ell_j$, $r_i \le r_j$. Based on $<_J$, an ordering of jobs that satisfies the adjacency property, we define a lexicographical ordering of the machines in the same manner as \secref{max-min}. Also as before, we may assume that $0 \le p_i \le 1$ by scaling down the processing time values.

\begin{remark}
\label{rem:Min-Max_double_interval}
Similar to the case of Max-Min allocations, if the input convex graph $H$ is inclusion-free, meaning that it satisfies the enclosure property as well as the adjacency property, then not only the neighbourhood of every machine is an interval of jobs, but also every job is adjacent to a consecutive set of machines. More precisely, given the ordering of jobs $<_J$ that satisfies the adjacency and enclosure property and assuming that the set of machines $M$ is ordered lexicographically based on $<_J$, then the neighbourhood of each job $J_j \in J$ forms an interval in the set $M$. This can be seen as a consequence of \obsref{double-interval}.
\end{remark}

\noindent \textbf{Hall's Condition:} Hall's condition needs to be slightly modified for the case of Min-Max allocations. For subset $J'$ of jobs let $w(J')$ and $w_s(J')$ be the sum of the processing times of all jobs and small\footnote{We use definitions for \emph{small} and \emph{big} jobs similar to the definitions for \emph{small} and \emph{big} items in \secref{rounding_santa}.} jobs in $J'$ respectively. Also let $N_{[H]}(J')$ (or $N(J')$ when the graph $H'$ is implied by context) denote the neighbourhood of the subset of jobs $J'$ in $M$. A necessary condition for having a maximum load which is at most $1$ is that for every subset $J'$ of machines  $w(J') \leq |N(J')|$. More generally, assume each machine $M_i$, $1 \leq i \leq m$, has a maximum \emph{allowable} load denoted by $a(M_i)$. Also, for a subset of machines $M'$, let $a(M')$ be the sum of allowable loads of all the machines in $M'$. Then, Hall's condition for the Min-Max allocation problem states that in order to have an assignment in which every machine has a load below its allowable maximum, it is necessary to have $w(J') \leq a(N(J'))$ for every subset of jobs $J'$. For the case of Max-Min allocations on inclusion-free graphs, we only need to check this condition for every interval of machines.

\begin{lemma}
\label{lem:min-max_Hall}
In order to check Hall's condition for Min-Max allocations in an inclusion-free graph $H$, it is enough to check the following condition:
\begin{align}
\label{cond:Min-Max_hall}
\forall \ \ [\ell, r\, ] \subseteq [1, \, m]: \ \  w(\cJ([\ell, \, r])) \leq \sum_{i = \ell}^{r} a(M_i)
\end{align}
in which $\cJ([\ell, \, r])$ denotes the set of jobs whose entire neighbourhood of machines falls in the interval $[\ell, \, r]$. We refer to \condref{Min-Max_hall} as Hall's condition for Min-Max allocations.

\end{lemma}

\begin{proof}
The proof mostly follows the proof of \lemref{hall} except for the fact that it uses the property mentioned in \remref{Min-Max_double_interval}. Assume we are given an inclusion-free convex graph $H$. As before, we show that  there exists a subset of the jobs $J'$ for which \emph{(i)} $N(J')$ is not an interval and therefore can be represented as the union of several maximal intervals, and \emph{(ii)} Hall's condition is violated if and only if there exists an interval of machines $[\ell,\, r]$ for which \condref{Min-Max_hall} is violated. 

\noindent
$\Rightarrow$: Assume that $w(J') > a(N(J'))$ and that there exist several maximal intervals of machines $\cN_1, \, \cN_2, \, \ldots, \, \cN_t$ whose union gives $N(J')$. Since each job is adjacent to an interval of machines by \remref{Min-Max_double_interval}, there exists a corresponding partition of $J'$ into subsets $J'_1, \, J'_2, \, \ldots, \, J'_t$, such that $N(J'_i) = \cN_i$. Since Hall's condition is violated, we have

\begin{displaymath}
a(N_1) + \, \ldots \, + a(N_t) = a(N(J')) < w(J') = w(J'_1) + \, \ldots \, +  w(J'_t )
\end{displaymath}

\noindent
Thus, there must exist an $i, 1 \leq i \leq t$, for which $a(\cN_i) = a(N(J'_i)) < w(J'_i )$. Let $\ell$ and $r$ be the leftmost and rightmost machines in $N(J'_i)$. We have that $a(N(J'_i)) = \sum_{i = \ell}^{r} a(M_i) < w(J'_i) \leq w(\cJ([\ell, \, r]))$  since $w(J'_i) \subseteq w(\cJ([\ell, \, r]))$.

\noindent
$\Leftarrow$: Assume that \condref{Min-Max_hall} is violated for an interval of machines $[\ell, \, r] \subseteq [1, \, m]$, meaning that $w(\cJ([\ell, \, r])) > \sum_{i = \ell}^{r} a(M_i)$. Let $J'$ be $\cJ([\ell, \, r])$. Then, $N(J') \subseteq [\ell, \, r]$ since otherwise it means there exists a job in $J' = \cJ([\ell, \, r])$ whose neighbourhood stretches beyond the interval $[\ell, \, r]$, contradicting the definition of $\cJ([\ell, \, r])$. Therefore $w(J') > \sum_{i = \ell}^{r} a(M_i) \geq a(N(J'))$. Therefore, Hall's condition is violated for the set $J'$.
\end{proof}

\subsection{Rounding the Instance}

Before running the algorithm on the instance, we round the processing times of jobs based on an input error parameter $k$. First of all, by properly scaling the values we assume every job has a processing time less than or equal to one. We say a job is \emph{small} if its processing time is less than or equal to $\frac{1}{k}$ after scaling, and we say a job is \emph{big} if its processing time is strictly larger than $\frac{1}{k}$. As opposed to the Max-Min case, we round \underline{down} the processing times for big jobs for the Min-Max case. Thus, if $p_j$, the processing time of job $J_j$,  is in the interval $[\frac{1}{k}(1+\frac{1}{k})^i, \frac{1}{k}(1+\frac{1}{k})^{i+1})$, then it is replaced by $\frac{1}{k}(1+\frac{1}{k})^{i}$.  Using this method, we obtain at most $C =  \lceil \frac{ \log k }{ \log ( 1+ \frac{1}{k} ) } \rceil$  distinct processing times for big jobs (each of which is more than $\frac{1}{k}$, the maximum processing time for small jobs). For $1 \leq \tau \leq C$, let $q_\tau = \frac{1}{k}(1+\frac{1}{k})^{i}$. Each $q_\tau$, $1 \leq \tau \leq C$, denotes the rounded processing time of category $\tau$ of big jobs. 

\subsubsection{The Algorithm for Inclusion-Free Convex Graphs}
The main theorem of this section is the following.

\begin{theorem}
\label{thm:finalJob}
Let $H$ be an instance of the problem before rounding with $n$ jobs and $m$ machines. Then, for $k \geq 4$ there exists a $(1 + \frac{4}{k} + \frac{3}{k^2})$-approximation algorithm for the Min-Max allocation problem on inclusion-free convex graphs with the running time of $\bigO{(m + n)mn^{2(C + 1)}}$ in which $C \leq k ^ {1.4}$.
\end{theorem}

To show this result, we borrow from the techniques used in \secref{max-min} extensively. We use the notions of \emph{$t$-assignment}, \emph{input vector}, \emph{assignment vector}, and \emph{right-aligned} assignments in the same sense as we did for the Max-Min allocations (jobs are allocated to machines here, instead of items to players). We modify the definition of input vectors slightly for Min-Max allocations below. 

\begin{definition}[$t$-assignment for the Min-Max case]
A $t$-assignment, for any $t \ge 0$, is a feasible assignment such that every machine $M_i$ receives a set of jobs $J^i \subseteq [\ell_{M_i}, \, r_{M_i}]$ with total processing time \underline{at most} $t$.
\end{definition}

\begin{definition}[Input Vector]
\label{def:input_vec_min-max}
For a given convex graph $H$ for an input instance with $\nu_\tau$ big jobs of processing time $q_\tau$ for $\tau = 1, \, 2, \, \ldots, \, C$ and small jobs of total processing time in the interval $[\frac{\nu_0}{k}, \, \frac{\nu_0 + 1}{k})$, an input vector is a configuration vector of the form $\nu(H) = \Nu = (\nu_0, \, \nu_1, \, \ldots, \, \nu_{C})$. 
\end{definition}

Then, we can prove the following \emph{alignment} lemma.

\begin{lemma}[The Alignment Lemma for Min-Max Allocations]\label{lem:alignment_min-max}
Suppose there exists a $1$-assignment for a given convex graph $H$. Then, there also exists a \underline{maximal} $(1 + \frac{1}{k})$-right-aligned assignment for $H$. Furthermore, the two assignments will have an identical assignment vector. 
\end{lemma}

\begin{proof}
We assign the jobs in two rounds and make use of the adjacency property. During the first round, the big jobs are aligned to the right in the same manner as \lemref{alignment}. In doing so, we do not change the assignment vector since the same number of big jobs as the $1$-assignment are allocated to the machines. During the second round, we assign the small items. The only difference between the assignment of jobs and the assignment of items is that we would prefer to allocate as many jobs as possible to a machine as long as the makespan is below the threshold of $1 + \frac{1}{k}$. Therefore, we pack maximal sets of small jobs and assign them to the machines. Due to the way the assignment vectors are defined, we still have not changed the vector after this phase.
\end{proof}

By the problem definition, we are not allowed to have stranded jobs in the Min-Max allocation problem. Therefore, every right-aligned assignment that we make is also a \emph{non-wasteful} right-aligned assignment. The following two facts hold for any non-wasteful right-aligned assignment of jobs to machines.

\begin{fact}
\label{fact:existence_min-max}
Whenever there exists an arbitrary $1$-assignment of jobs to machines for an instance of the problem, there also exists a restricted ($1 + \epsilon$)-assignment, in the sense that all the jobs assigned are restricted to be aligned to the right.
\end{fact}

\begin{fact}
\label{fact:retrieve_min-max}
Given that a $1$-assignment exists for an instance of the problem, there also exists a polynomial-time algorithm that, given the input vectors for every step (partial assignment) of a ($1 + \epsilon$)-assignment, reconstructs the remainder graph of every step in such a way that the total processing time of the jobs in each reconstructed remainder graphs is

\begin{itemize}
\item the same as that of the corresponding remainder graph in the original (and unknown) $1$-assignment for every category $\tau = 1, \, 2, \, \ldots, \, C$ of big jobs and,
\item only a small fraction $\epsilon$ less than its counterpart in the original remainder graph for small jobs.
\end{itemize}
\end{fact}

While \lemref{alignment_min-max} verifies \factref{existence_min-max} for non-wasteful right-aligned assignments, the following lemma ensures that \factref{retrieve_min-max} also holds for these types of assignments. \\

\begin{lemma}
\label{lem:non-wasteful_assign_min-max}
Assume there exists a non-wasteful right-aligned $1$-assignment $\bcA = (J^1, \, J^2, \, \ldots, \, J^m)$ for an instance of the problem. Also assume we are given the original inclusion-free convex graph $H$ and the input vectors for partial assignments $\bcA^j = (J^j, \, J^{j + 1}, \, \ldots, \, J^m)$ for $j = m - 1, \, m - 2, \, \ldots, \, 1$, but not the partial assignments themselves. There exists a polynomial time algorithm that can reconstruct the remainder graphs for every $\bcA^j$ in such a way that the total processing time of jobs in each reconstructed graph is at most $\epsilon = \frac{2}{k}$ less than the total processing time of jobs in the actual corresponding remainder graph for $\bcA^j$.
\end{lemma}

\begin{proof}
The proof uses a left-to-right sweep again. We do not provide the details here since it is identical to the proof of \lemref{non-wasteful_assign}, except for the way small jobs are chosen in the reconstruction process. For small jobs, we select a \underline{minimal} set of left-most jobs whose total processing time in greater than $\frac{\nu_0 - 1}{k}$ and add them to the remainder graph. Since the processing time of each small job is less than $\frac{1}{k}$, this ensures that the sum of processing times of the jobs we choose in the reconstructed remainder graph is less than or equal to that of the original remainder graph, but the deficit is not more than $\frac{1}{k}$. The rest of the proof follows the proof of \lemref{non-wasteful_assign}.
\end{proof}

\noindent
\textbf{Algorithm Assignment (for the Min-Max Problem on Inclusion-Free Convex Instances):} Algorithm Assignment for the Min-Max allocation problem is a dynamic programming algorithm similar to the algorithm given in \secref{max-min}. Again we make use of two utility functions, namely the reconstruction function \textsf{retrieve} and the boolean function \textsf{feasible}. Function \textsf{retrieve} reconstructs a remainder graph based on an input vector and a player index which potentially has slightly less total processing time than the original remainder graph of a partial assignment, and function \textsf{feasible} checks whether a given assignment of jobs is feasible in a remainder graph and does not leave any stranded jobs behind. Using these functions, the algorithm fills in the entries of a $m \times n^{C + 1}$ table $T$. Each entry $(i, \, \Nu)$ of the table has two fields: 

\begin{enumerate}
\item $T(i, \, \Nu).$\textsf{bit}: this binary value is set to $1$ if the input vector $\Nu$ is marked for machine $i$, and set to $0$ otherwise.
\item $T(i, \, \Nu).$\textsf{ptr}: contains either \textsf{NULL} or a reference to another input vector $\Nu'$ in the previous row, intended for use in the backward phase of the algorithm where an actual assignment of the jobs to the machines is retrieved from the table $T$.
\end{enumerate}

The algorithm is as follows.

\begin{algorithm}
 \KwData{rounded instance $H$.}
 \KwResult{returns either a ($1 + \frac{3}{k}$)-assignment for all machines in the rounded instance, or report failure.}
 \For{every vector $\Nu$ in $\cV$ such that $\Nu \leq \Nu_{in}$} {
  $\hat{H} \leftarrow \textsf{retrieve}(\Nu, \, m - 1)$ \\
  $\hat{J} \leftarrow$ the set of jobs in $\hat{H} \, ; \, J^m \leftarrow J \, \backslash \, \hat{J}$ \\
  \If{$\textsf{feasible}(H, \, J^m)$} {
   $T(m, \, \Nu).\textsf{bit} \leftarrow 1$ \\ 
   \If{$T(m, \, \Nu).\textsf{ptr} == \textsf{NULL}$} {
    $T(m, \, \Nu).\textsf{ptr} \leftarrow \Nu_{in}$
   }
  }
 }
 \For{$i \leftarrow m - 1$ downto $1$}{
  \For{every vector $\Nu$ in $\cV$ such that $\Nu \leq \Nu_{in}$} {
   \For{every vector $\Nu'$ in $\cV$ such that $\Nu' \geq \Nu$ and $T(i + 1, \, \Nu').\textsf{bit} == 1$} {
    $\hat{H} \leftarrow \textsf{retrieve}(\Nu, \, i - 1) \, ; \, \hat{H}' \leftarrow \textsf{retrieve}(\Nu', \, i) $ \\
    $\hat{J} \leftarrow$ the set of jobs in $\hat{H}$ \, ; \, $\hat{J}' \leftarrow$ the set of jobs in $\hat{H}' \, ; \, J^i \leftarrow \hat{J}' \, \backslash \, \hat{J}$ \\
    \If{$\textsf{feasible}(\hat{H}', \, J^i)$} {
     $T(i, \, \Nu).\textsf{bit} \leftarrow 1$ \\ 
     \If{$T(i, \, \Nu).\textsf{ptr} == \textsf{NULL}$} {
      $T(i, \, \Nu).\textsf{ptr} \leftarrow \Nu'$
     }
    }
   }
  } 
 }
 \caption{Algorithm Assignment for Inclusion-Free Instances (Forward Phase)}
\end{algorithm}

\begin{lemma}\label{lem:finalJob}
Let $H$ be a rounded instance of the problem with $n$ jobs and $m$ machines. Algorithm Assignment assigns (in polynomial time) to each machine a set of jobs with value at most $1 + \frac{3}{k}$ if a $1$-assignment exists for all machines. The algorithm returns failure if no ($1 + \frac{1}{k}$)-assignment exists for $H$. 
\end{lemma}

\begin{proof}
By \lemref{alignment_min-max}, a $(1 + \frac{1}{k})$-assignment exists whenever there is a $1$-assignment for the problem instance. Using the retrieve function, we may guess the remainder graph of any of the partial assignments to have a deficiency in total processing time which is no more than $\frac{2}{k}$ when compared to the original remainder graph. This together with the error caused by the alignment imply the total sum of processing times of the jobs assigned to each machine is at most $\frac{3}{k}$ more than the same value in the $1$-assignment (if a successful assignment is identified). It is also straightforward to see that whenever there exists no $(1 + \frac{1}{k})$-assignment for the instance, the algorithm fails to recover a feasible assignment for a machine, resulting in a failure report. This completes the correctness proof of the Algorithm Assignment.
\end{proof}

In the last part, we prove \thmref{finalJob} using the lemmas in this section. 

\begin{proof}[Proof of \thmref{finalJob}]
First we discuss the approximation guarantee. Through alignment and reconstruction of the remainder graphs, we may pack an extra $\frac{3}{k}$ units of processing time in the bundle of jobs assigned to a machine. This extra value is magnified by a multiplicative factor of $1 + \frac{1}{k}$ in the rounded instance. Therefore, the approximation guarantee would be $(1 + \frac{3}{k}) \cdot (1 + \frac{1}{k}) = 1 + \frac{4}{k} + \frac{3}{k^2}$. We now discuss the running time of the algorithm. The two inner loops of the Algorithm Assignment take a combined total time of $n^{2(C + 1)}$ while the outer loop runs for each machine for a total of $m$ times. Retrieving the remainder graph can be done in $\bigO{m + n}$, hence the running time is $\bigO{(m + n)n^{2(C + 1)}}$ in total. The correctness is already shown in \lemref{finalJob}.
\end{proof} 

\section{Conclusion and Future Work}

In all instances of the problem considered in this paper, a proper ordering has played an important role. We conjecture that a dynamic programming algorithm similar to the one used in this paper provides a PTAS for the case of margined-inclusive intervals. However, we do not know a dichotomy classification for the instances of the problem that admit a PTAS. We ask for a dichotomy of the following form. If $H$ belongs to class $X$ of bipartite graphs, then there is a PTAS for Max-Min allocation problem. Otherwise, there is no PTAS.

\bibliographystyle{abbrv}

\bibliography{santa_min_new}

\end{document}

%% file: preamble.tex
\usepackage[]{algorithm2e}
\usepackage[scale=0.8]{geometry}
\usepackage[active]{srcltx}
\usepackage{amsmath}
\usepackage{amssymb}
\usepackage{amsfonts}
\usepackage{amsthm}
\usepackage{bm}

\usepackage{microtype}
\usepackage{graphicx}
\usepackage{float}
\usepackage[T1]{fontenc}
\usepackage[utf8]{inputenc}
\usepackage{authblk}
\usepackage{caption}
\usepackage{subcaption}
\usepackage[table]{xcolor}
\usepackage{tikz}
\usetikzlibrary{plotmarks}
\usetikzlibrary{decorations.markings}
\usepackage{chngpage}
\usepackage[american]{babel}
\usepackage{graphicx}
\usepackage{color}
\usepackage{multirow}
\usepackage[margin=0in,font=small,labelfont=bf,indention=0cm]{caption}  
\usepackage{latexsym}
\usepackage{setspace}   
\usepackage{framed}

\newcommand{\NOTEB}[2]{$^{\textcolor{blue}\clubsuit}$\marginpar{\setstretch{0.43}\textcolor{blue}{\bf\tiny #1: }\textcolor{blue}{\bf\tiny #2}}}

\renewcommand{\NOTEB}[2]{}

\renewcommand{\phi}{r}

\usepackage[numbers,sort&compress,longnamesfirst,sectionbib]{natbib}
\floatstyle{ruled}
\newfloat{algorithm}{htbp}{loa}
\floatname{algorithm}{Algorithm}

\renewcommand{\Xi}{\xi}

\newtheorem{thm}{Theorem}  
\newtheorem{fact}[thm]{Fact}
\newtheorem{lem}[thm]{Lemma}

\newtheorem{cor}[thm]{Corollary}

\newtheorem{pro}[thm]{Proposition}

\newtheorem{exm}[thm]{Example}
\newtheorem{definition}{Definition}
\newtheorem{remark}{Remark}
\newtheorem{lemma}{Lemma}
\newtheorem{theorem}{Theorem}
\newtheorem{observation}{Observation}
\newtheorem{claim}{Claim}

\numberwithin{thm}{section}
\numberwithin{equation}{section}

\newcommand{\thmref}[1]{Theorem~\ref{thm:#1}}

\newcommand{\lemref}[1]{Lemma~\ref{lem:#1}}
\newcommand{\lemrefs}[2]{Lemmas~\ref{lem:#1} and~\ref{lem:#2}}

\newcommand{\obsref}[1]{Observation~\ref{obs:#1}}
\newcommand{\clmref}[1]{Claim~\ref{clm:#1}}
\newcommand{\defref}[1]{Definition~\ref{def:#1}}
\newcommand{\defrefs}[2]{Definitions~\ref{def:#1} and~\ref{def:#2}}

\newcommand{\figref}[1]{Figure~\ref{fig:#1}}

\newcommand{\secref}[1]{Section~\ref{sec:#1}}

\newcommand{\remref}[1]{Remark~\ref{rem:#1}}

\newcommand{\exmref}[1]{Example~\ref{exm:#1}}
\newcommand{\factref}[1]{Fact~\ref{fact:#1}}
\newcommand{\condref}[1]{Condition~\ref{cond:#1}}

\newcommand{\bQ}{\ensuremath{\mathbb{Q}}}

\newcommand{\cA}{\ensuremath{\mathcal{A}}}

\newcommand{\cJ}{\ensuremath{\mathcal{J}}}

\newcommand{\cN}{\ensuremath{\mathcal{N}}}

\newcommand{\cV}{\ensuremath{\mathcal{V}}}

\newcommand{\bcA}{\ensuremath{\boldsymbol{\mathcal{A}}}}

\newcommand{\Alpha}{\boldsymbol\alpha}
\newcommand{\Nu}{\boldsymbol\nu}

\newcommand{\NPH}{\ensuremath{\mathcal{NP}}-hard}

\newcommand{\bigO}[1]{\ensuremath{\mathcal{O} \left( #1 \right)}}

\renewcommand{\epsilon}{\ensuremath{\varepsilon}}

\let\oldsqrt\sqrt
\def\hksqrt{\mathpalette\DHLhksqrt}
\def\DHLhksqrt#1#2{\setbox0=\hbox{$#1\oldsqrt{#2\,}$}\dimen0=\ht0
   \advance\dimen0-0.2\ht0
   \setbox2=\hbox{\vrule height\ht0 depth -\dimen0}%
   {\box0\lower0.4pt\box2}}
\renewcommand\sqrt\hksqrt
\renewcommand{\leq}{\leqslant}
\renewcommand{\geq}{\geqslant}
\renewcommand{\le}{\leqslant}
\renewcommand{\ge}{\geqslant}
\renewcommand{\epsilon}{\varepsilon}

\newcount\minute \newcount\hour \newcount\hourMins
\def\now{\minute=\time \hour=\time \divide \hour by 60 \hourMins=\hour \multiply\hourMins by 60
  \advance\minute by -\hourMins \zeroPadTwo{\the\hour}:\zeroPadTwo{\the\minute}}

\def\today{\the\year-\zeroPadTwo{\the\month}-\zeroPadTwo{\the\day}}
\def\zeroPadTwo#1{\ifnum #1<10 0\fi #1}

\allowdisplaybreaks[1]   

\usepackage[T1]{fontenc}
\usepackage{lmodern}